\documentclass[preprint,a4paper]{elsarticle}

\usepackage[utf8]{inputenc}
\usepackage{amsmath}
\usepackage{stmaryrd} 
\usepackage{amsfonts}
\usepackage{amssymb}
\usepackage{graphicx}
\usepackage{soul}

\usepackage{hyperref}
\hypersetup{hidelinks}

\usepackage{color}
\usepackage{setspace}

\usepackage[colorinlistoftodos]{todonotes}
\usepackage{enumerate}

\usepackage{float}

\newcommand{\hcc}{\mathrel{\scalebox{0.8}{$\curlyveeuparrow$}}}

\newtheorem{theorem}{Theorem}
\newtheorem{definition}[theorem]{Definition}
\newtheorem{assumption}[theorem]{Assumption}

\newtheorem{lemma}[theorem]{Lemma}
\newtheorem{corollary}[theorem]{Corollary}
\newproof{proof}{Proof}

\author[1]{Zsigmond Benkő\fnref{fn1}}
\author[1,3]{Ádám Zlatniczki\fnref{fn1}}
\author[1]{Marcell Stippinger\fnref{fn1}}
\author[4]{Dániel Fabó}
\author[4]{András Sólyom}
\author[5,6]{Loránd Erőss}
\author[1,3,7]{András Telcs\fnref{fn2}}
\author[1,8]{Zoltán Somogyvári\fnref{fn2}\corref{cor1}}

\fntext[fn1]{Equal contribution as first authors.}
\fntext[fn2]{Equal contribution as last authors.}
\cortext[cor1]{Corresponding author at: Department of Computational Sciences, Institute for Particle and Nuclear Physics, HUN-REN Wigner Research Centre for Physics, Budapest H-1121, Hungary. E-mail: \nolinkurl{somogyvari.zoltan@wigner.hun-ren.hu} (Z. Somogyvári).}
\affiliation[1]{organization={Department of Computational Sciences, Institute for Particle and Nuclear Physics, HUN-REN Wigner Research Centre for Physics},
addressline={Konkoly-Thege Miklós út 29-33.},
postcode={H-1121},
city={Budapest},
country={Hungary}}
\affiliation[3]{organization={Department of Computer Science and Information Theory, Faculty of Electrical Engineering and Informatics, Budapest University of Technology and Economics},
addressline={Magyar tudósok körútja 2.},
postcode={H-1111},
city={Budapest},
country={Hungary}}
\affiliation[4]{organization={Epilepsy Unit, Department of Neurology, National Institute of Mental Health, Neurology and Neurosurgery},
addressline={Amerikai út 57.},
postcode={H-1145},
city={Budapest},
country={Hungary}}
\affiliation[5]{organization={Department of Functional Neurosurgery, National Institute of Mental Health, Neurology and Neurosurgery},
addressline={Amerikai út 57.},
postcode={H-1145},
city={Budapest},
country={Hungary}}
\affiliation[6]{organization={Faculty of Information Technology and Bionics, Péter Pázmány Catholic University},
addressline={Práter u. 50/A},
postcode={H-1083},
city={Budapest},
country={Hungary}}
\affiliation[7]{organization={Department of Quantitative Methods, Faculty of Business and Economics, University of Pannonia},
addressline={Egyetem u. 10.},
postcode={H-8200},
city={Veszprém},
country={Hungary}}
\affiliation[8]{organization={Axoncord LLC},
addressline={Dunakeszi utca 23. Fsz. 1. ajtó},
postcode={H-1048},
city={Budapest},
country={Hungary}}

\begin{document}
\title{Complete Inference of Causal Relations between Dynamical Systems}
\date{January 16, 2024}


\begin{abstract}

From ancient philosophers to modern economists, biologists, and other researchers, there has been a continuous effort to unveil causal relations. The most formidable challenge lies in deducing the nature of the causal relationship: whether it is unidirectional, bidirectional, or merely apparent — implied by an unobserved common cause.

While modern technology equips us with tools to collect data from intricate systems such as the planet's ecosystem or the human brain, comprehending their functioning requires the identification and differentiation of causal relationships among the components, all without external interventions.

In this context, we introduce a novel method capable of distinguishing and assigning probabilities to the presence of all potential causal relations between two or more time series within dynamical systems. The efficacy of this method is verified using synthetic datasets and applied to EEG (electroencephalographic) data recorded from epileptic patients.

Given the universal applicability of our method, it holds promise for diverse scientific fields.

\end{abstract}

\begin{keyword}
causal analysis \sep Bayesian inference \sep time series \sep Takens' theorem \sep topological embedding \sep intrinsic dimension \sep epileptic focus localisation \sep EEG analysis
\end{keyword}

\maketitle

Causality stands as one, if not the most fundamental pillar of science. Yet, the identification of causal relations within deterministic dynamical systems, relying solely on observations without interventions, remains a formidable challenge. Despite numerous proposed methods, none has yielded an exact and comprehensive solution.

{\bf Predictive causality} constitutes one branch of standard causality analysis tools, originating from its foundational principle articulated by Norbert Wiener \cite{Wiener56}.
In this framework, one examines two time series in conjunction with a model that forecasts a given time series based on its recent history. If the incorporation of the recent past of another time series improves the prediction of the first, one asserts that, in the Wiener sense, the included time series is a cause of the first.
This principle, initially implemented by Granger \cite{Granger1969}, has undergone subsequent extensions to encompass nonlinear models and non-parametric methods, such as transfer entropy \cite{Schreiber2008, MartinPalus2008, Vicente2011transfer}.

The Granger method, rooted in the predictive causality principle and tailored for analyzing stochastic time series, has gained widespread popularity across various scientific domains. Simultaneously, there are several indications that, in the context of deterministic dynamical systems, it may yield inaccurate results~\cite{Lusch2016}.

Judea Pearl \cite{Pearl2000} introduced an axiomatic approach to infer causality in networks of random variables in terms of directed cyclic graphs. This method can infer the direction of connections but its applicability is limited to specific network structures and not suitable for scenarios involving only two variables.  

{\bf Topological causality} represents a different approach to causality analysis, focusing to interacting deterministic dynamical systems. These methods become widely known by the work of George Sugihara \cite{Sugihara2012}, further developed by Tajima \cite{tajima2015untangling} and refined by Harnack \cite{harnack2017topological} and many others. Sugihara's convergent cross mapping (CCM) technique is rooted in Takens' embedding theorem \cite{Takens1981}.
Takens' theorem states that the topology and all degrees of freedom within a dynamical system can be reconstructed from time series observations through time delay embedding. The CCM method's pivotal concept lies in the notion that the consequence encompasses an observation (in Takens' sense, a lower-dimensional smooth function of the state) of the cause \cite{Stark1999}. Thus, all degrees of freedom of the cause can be reconstructed from the consequence.

While Sugihara's CCM method is able to hint bidirectional causal relations, distinguishing a direct causal link from a hidden common cause remains a formidable challenge. Both Sugihara \cite{Sugihara2012} and Harnack \cite{harnack2017topological} acknowledge that observed correlation without inferred direct causation signifies a common driver. However, it is important to note that a common driver does not necessarily result in linear correlation between the two driven systems. Consequently, the CCM method can detect directed and bidirectional causal relationships but may struggle to infer all cases with hidden common cause included, as we demonstrated with three nonlinearly coupled logistic maps (see Extended Data \autoref{fig:CommonCauseTest}). Further insights into its applicability can be found in \citep{mccracken2014convergent,cobey2016limits,monster2017causal}.

Recurrence maps have been applied in causality detection too. Hirata et al.\ \cite{Hirata2010} inferred the presence of a common cause by rejecting both independence and direct dependence based on recurrence maps. However, the method falls short in providing quantitative detection of all types of causal relations, relying on classical hypothesis tests that offer only unidirectional implication. Consequently, while common cause can be detected by Hirata's method \cite{Hirata2010}, direct causality and independence is discerned only on the ``cannot be rejected'' branch of their tests.

The method we propose originates from Takens' topological theory of dynamical systems \cite{Takens1981} and draws insights from information theory and dimension theory \cite{Renyi1959, pincus1991approximate, Geiger2019}. There are a few earlier papers that utilize dimension as a measure of dependency, without full discovery of possible causal relations, e.g., \cite{romano2016,Quiroga2000}. In the following sections, we introduce our method and its applications to various types of time series. We showcase its novelty and uniqueness in its unified nature, demonstrating its capability to detect and assign probabilities to all types of causal relations, notably excelling in detecting hidden common causes in dynamical systems.

Our method provides complete detection and distinction of  all possible causal relations within two deterministic dynamical systems.\footnote{The earlier version of this paper was made available in the manuscript \cite{benko2020complete}. To our best knowledge the proposed method was the first exact one which is capable to detect all kind of causal relationships, including hidden common cause.}

\section*{Inferring causality from manifold dimensions}

Inferring causal relations through manifold dimensions is feasible and roots in Takens' theorem. When provided with a time series, we can perform time delay embedding and measure/estimate the dimensionality of the reconstructed dynamics.
Takens' theorem guarantees that the dimension of the embedded manifold aligns with the dimension of the attractor in the original state-space, remaining invariant against the observation function and dimension of the embedding space (provided it is sufficiently high dimensional).

Various dimensionality notions have been introduced to capture essential characteristics of attractors in dynamical systems. These explorations have led to the adoption of diverse fractal-type dimensionality notations, including Rényi information dimension, correlation dimension, and various intrinsic dimensions \cite{Renyi1959, mandelbrot1979fractals, Grassberger1983, Levina2005, Szepesvari2007}.

As an illustrative example of the potential use of manifold dimensions for causal inference, consider two simulated, unidirectionally coupled dynamical systems and their reconstructed attractors depicted in \autoref{fig:logmap} A. It is evident that the dimensions of the two embedded manifolds differ. Specifically, the dimension of the consequence is greater than that of the cause, signifying that the consequence `contains' the degrees of freedom of the cause. Detecting this difference in dimensionality between cause and consequence opens up the opportunity to establish a new analytical method capable of distinguishing all forms of causal relations.

The key variable in our framework is the joint dimension, denoted as $D_J$, which involves two simultaneous time series. To obtain $D_J$, one can create the direct product of the two embedded spaces 
and measure the intrinsic dimension within the resulting point cloud (see \ref{si:dim-cause} for details).
It is important to note that the asymmetry of manifold dimensions in the subspaces alone does not necessarily imply causation. However, the additional information offered by $D_J$ is adequate to ascertain the type of causal relation between the two systems.

\subsection*{Causal and dimensional relations}

In the case of two independent dynamical systems, the joint dimension is equal to the sum of the dimensions of the two independent systems. However, any interdependence renders the manifold dimensions sub-additive.

Given that the cause can be reconstructed from the consequence, the information content of the cause is already present in the embedded manifold of the consequence. Consequently, in the case of unidirectional coupling, the joint dimension is equivalent to the dimension of the consequence.

It's noteworthy that the dimension of the driven dynamical system (consequence) is always greater than or equal to the dimension of the driver dynamical system (cause). The dimensions of the two manifolds decisively determine the direction of the potential causal effect: only the lower-dimensional system can have a unidirectional causal effect on the higher-dimensional one.

A special case of driving arises when there is bidirectional coupling (sometimes called  circular or  mutual coupling). According to Takens' theorem, a bidirectional case implies the existence of homeomorphisms in both directions, making the two manifolds topologically equivalent. As a consequence, the joint dimension is equal to the dimension of both time series. However, observing that all the three dimensions are equal, there are still two possible scenarios. Either we have true bidirectional coupling between the systems, or we have the so-called generalized synchrony, indicating full determinism, in fact, meaning that we are observing the same dynamic system twice.

If the joint dimension is less than the sum of the single dimensions but not equal to either of them, it signals the presence of a hidden common cause without a direct causal effect between the two time series.

In summary, the relationship between the dimensions of the two systems and the joint dimension distinguishes the four possible causal scenarios as follows:
\smallskip

Independent case
\begin{equation}
	X \perp Y \iff D_X + D_Y = D_J
\label{indep}
\end{equation}

Unidirectional cases
\begin{subequations}
\label{uni}
\begin{gather}
X \rightarrow Y \iff  D_X < D_Y = D_J \\
X \leftarrow Y \iff  D_Y < D_X = D_J
\end{gather}%
\end{subequations}

Bidirectional case
\begin{equation}
X \leftrightarrow Y \iff D_X = D_Y = D_J
\label{cir}
\end{equation}

Common cause
\begin{equation}
	X \hcc Y \iff  \max(D_X , D_Y) < D_J <  D_X + D_Y
\label{hcc}
\end{equation}

The implications in Eqs.\ \eqref{indep}--\eqref{cir} straightforwardly connect causal relations to dimensional relations. The case of a common cause and the reasoning in the opposite direction follow from discretizing state variables; otherwise, additional assumptions are required (see \ref{si:dim-cause}). Therefore, to infer causality, we estimated these dimensions using the method proposed by Farahmand \cite{Szepesvari2007}.

As we only have estimates for manifold dimensions ($\bar{d}$) based on finite datasets, we cannot demonstrate exact equalities in Eqs.\ \eqref{indep}--\eqref{hcc}, but we assign probabilities to these causal cases.
These posterior probabilities of the dimensional relationship on the r.h.s.\ of Eqs.\ \eqref{indep}--\eqref{hcc} are inferred by comparing the estimates of joint dimension to the dimensions of the individual embedded time series.

The inference is demonstrated by the following scheme. Let $A$ be the causal relation, which unequivocally determines the relation between dimensions according to Eqs.\ \eqref{indep}--\eqref{hcc}. Also, let $\bar{d}$ denote the observed dimension vector, formed by the $\left( \bar{d}_X, \bar{d}_Y, \bar{d}_J, \bar{d}_Z \right)$ estimations of $(D_X, D_Y, D_J, D_X+D_Y)$ respectively. First, we apply Bayes' theorem

\begin{equation*}
P(A|\bar{d}) = \frac{p\left(\bar{d}|A\right)}{p(\bar{d})}P(A),
\end{equation*}
then assume a non-informative, uniform prior over the possible causal relations $ \\ ( {A \in \nobreak\{} {\leftarrow\nobreak,} {\leftrightarrow\nobreak,} {\rightarrow\nobreak,} {\hcc\nobreak,} {\perp\nobreak\}})$ and then calculate the conditional likelihood of the observed dimensions as  
\begin{equation*}
p(\bar{d}|A) %
= \int p(\bar{d}|w) 
dP(w|A),
\end{equation*}
where the integrand $w$ ranges over the admissible dimension combinations, determined by $dP(\cdot | A)$.

To derive numerical formulas, we need the distribution of the measurement of manifold dimensions. We assume homogeneous sampling of the points of the manifold.  Given that the  local neighborhoods of sample points might overlap, the estimate of the joint probability distribution of the dimensions needs effective sample size correction (see \ref{Sapp}).

Consequently, the average of local dimension estimates $p(\bar{d}|w)$, i.e., the conditional measurement of the manifold dimensions, is a multivariate normal distribution centered at the mean vector $w$ where the covariance matrix $s$ accounts for the correlation across manifolds mediated by the shared information. 
Thus, we have
\begin{equation*}
p(\bar{d}|A) = \int \varphi_{w,s}(\bar{d}) \; dP_{D,\Sigma|A}(w,s|A).
\end{equation*}
The covariance $s$ is approximated by the sample covariance of local dimensions. For further details see \ref{si:probabilities}.

Dimension estimates are sensitive to parameter choices, such as embedding dimension (see \ref{si:workflow}), and are influenced by unavoidable biases. The impact of bias caused by the embedding dimension can be mitigated if the manifolds are embedded into spaces with the same embedding dimensions.
Following this principle, instead of summing the $D_X$ and $D_Y$ individual dimensions, the independent case was represented by the dimension of a manifold $Z$ that is constructed by combining $X$ and the time-permuted (independently manipulated) $Y^*$ manifolds. 

To align the embedding dimensions of $J$ and $Z$ with the embedding dimensions of $X$ and $Y$, instead of the direct product of the embedding spaces, $J' = aX + Y$ and $Z' = aX + Y^*$, embeddings of the linear combinations were used, where $a$ is a properly chosen irrational number. It's important to note that this way, $J'$ and $Z'$ are based on generic observation functions of the original systems, making them topologically equivalent to the original system in almost all cases (for proof, see \ref{si:dim-cause}).

We assume in the whole sequel that the mapping from the cause to the consequence is a generic observation of the cause in the sense of Takens' theorem. Otherwise, for instance, if the cause consists of two independent subsystems and only one (say the first one) affects the consequence, then the dimension of the  joint and the consequence manifold are different; in fact we face a common cause situation in which the cause is also consequence of the first part of itself. 

Certain situations of system interactions require additional consideration, including the case of full determinism, the coexistence of a direct connection and a hidden common driver, and the transitivity of causal relations.

In the case of full determinism, also known as generalized synchrony, the caused time series are entirely determined by the cause. In this scenario, considering them as different systems is not meaningful; we essentially have two copies/observations of the same system (potentially with a time delay). Thus, this specific unidirectional coupling cannot be distinguished from the bidirectionally coupled case, as the reconstructed topologies are equivalent in both scenarios.

There are specific cases where the common driver remains concealed. If the common cause coexists with a direct connection in at least one direction, the direct connection ensures that all the information from the cause is present in the joint; as a result, the direct connection(s) will be detectable, while the common driver remains hidden.

In theory, causality is a transitive relation. Consequently, indirect causal relations (occurring through a chain of direct connections) should be identified as direct causal relations as well. However, in real-world applications, additive noise can lead to non-transitive directed causal relations among multiple systems.

\begin{figure}
	\centering
	\includegraphics[scale=1.6]{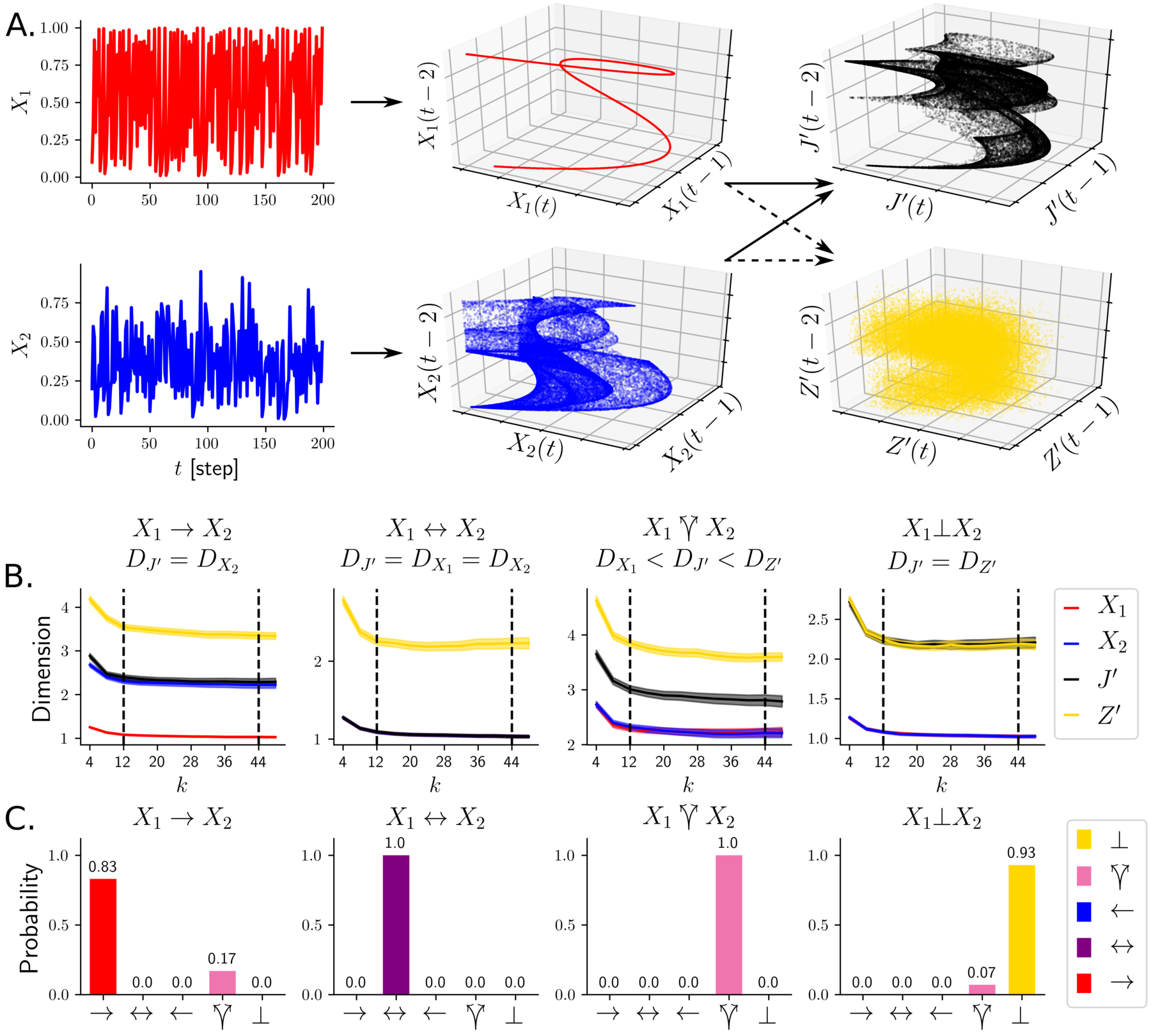}
	\caption{\textit{The workflow and testing of our method on coupled logistic map systems (see Eq.\ \eqref{eq:logist})}.
    (A)  The state spaces of the systems are reconstructed by time delay embedding of the two time series $X_1$ (red) and $X_2$ (blue), resulting in the red and blue manifolds. Then, the joint of the two datasets, $J'$, and their time-shuffled version $Z'$ are also embedded, resulting in a reconstruction of the joint state space of the two subsystems (black manifold) and their independent joint (yellow).
    On (B,C) the test of our method on the four simulated examples of five possible causal interactions (one of the unidirectional: $X_1 \rightarrow X_2$, bidirectional: $X_1 \leftrightarrow X_2$, unidirectional backwards: $X_1 \leftarrow X_2$, common cause: $X_1 \protect\hcc X_2$, independence: $X_1 \perp X_2$) are demonstrated.
    (B) The intrinsic dimensionality of each manifold is estimated for different neighborhood sizes $k$. The plateau of dimension-estimates identifies where the estimates can be considered reliable (between dashed lines). Note the match between the actual causal and dimensional relationships: the dimension of the joint manifold ($J$) relative to the others.
    (C) Posterior probabilities of the possible causal relationships. The method correctly assigned the highest probability to the actual causal relation in each case.}
	\label{fig:logmap}  
\end{figure}

\section*{Results}

We validate our method using data from three simulated dynamical systems where ground truths are known. Following this verification, the method is also applied to EEG data from epileptic patients. The three simulated dynamical systems have fundamentally different dynamics, presenting diverse challenges that causal analysis must overcome: the coupled logistic maps are discrete-time chaotic dynamical systems with no significant temporal autocorrelation; in contrast, the coupled Lorenz systems are defined in continuous time and exhibit smooth temporal autocorrelation, while the Hindmarsh--Rose models display steep spikes and quasi-periodic behavior.

\subsection*{Logistic maps} We simulated systems of three coupled logistic maps with various connectivity patterns: unidirectional, bidirectional, independent and two uncoupled maps driven by a third unobserved one:
\begin{equation}
	x_i[t+1] = r x_i[t] (1 - x_i[t]-\sum_{j\neq i} \beta_{i,j} x_j[t]),
\label{eq:logist}
\end{equation} 
where $i \in \lbrace 1, 2, 3 \rbrace$ for the three variables, $r=3.99$ and $\beta_{i,j}$ are coupling coefficients (values for the different simulated couplings are given in \ref{si:results}).
 
The performance of the method was evaluated on simulated coupled logistic maps. Our method was able to reveal the original coupling pattern between the observed logistic maps for all cases, in particular it was able to detect the existence of the hidden common cause observing only the two affected logistic maps (\autoref{fig:logmap}). For further details see \ref{si:results}.

\subsection*{Lorenz systems} We tested our method on differently coupled Lorenz systems \cite{Lorenz1963} (\autoref{fig:lorenz}). 
\begin{figure}
	\centering
	\includegraphics[scale=1.8]{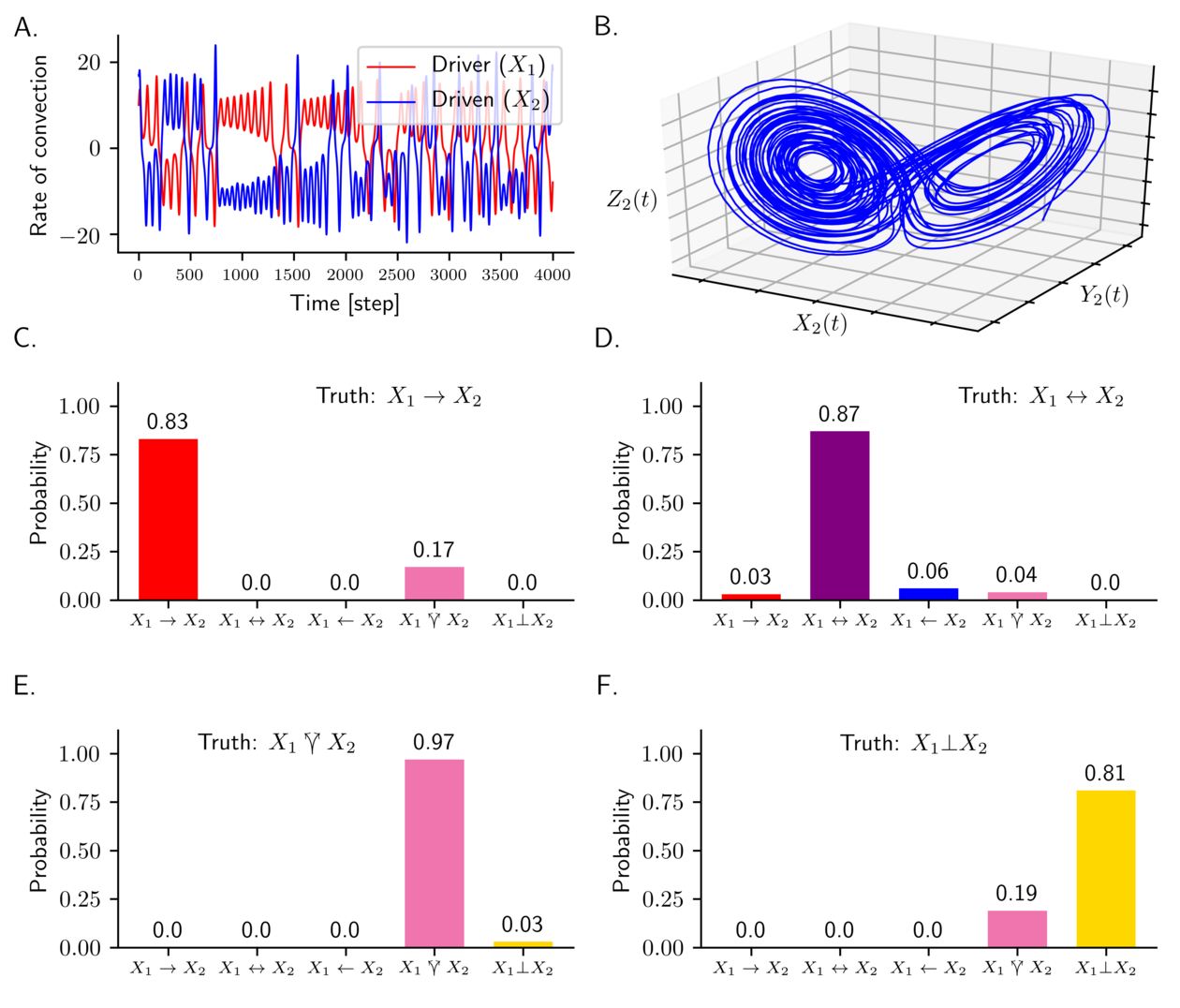}
	\caption{\textit{Testing our method on coupled Lorenz systems.}
    (A) The $X$ variables of two Lorenz systems, represented as time series. There is unidirectional coupling from the first system (the driver, red) to the second (the driven, blue). The description of the Lorenz systems is given in the Appendix.
    (B) The attractor of the driven system is only slightly perturbed by the driver.
    (C)--(F) Model probabilities nicely match with the ground truth for the different couplings. (Color code matches \autoref{fig:logmap} C.)
	}
	\label{fig:lorenz}
\end{figure}

We conclude that in each case our method detects the proper causal relationship with high confidence. For further details see \ref{si:results}.

\subsection*{Hindmarsh--Rose systems} We analyze causal relationships between coupled Hindmarsh--Rose systems which were originally proposed to model spiking or bursting neuronal activity \cite{Hindmarsh-Rose1984}. In our simulations we use two electrically coupled neurons where coupling is achieved through the membrane potential, as proposed in \cite{Xia2005coupled-hindmarsh-rose}. Analysis of such time series is quite difficult due to their sharp dynamical changes.  

\begin{figure}
	\centering
	\includegraphics[scale=2]{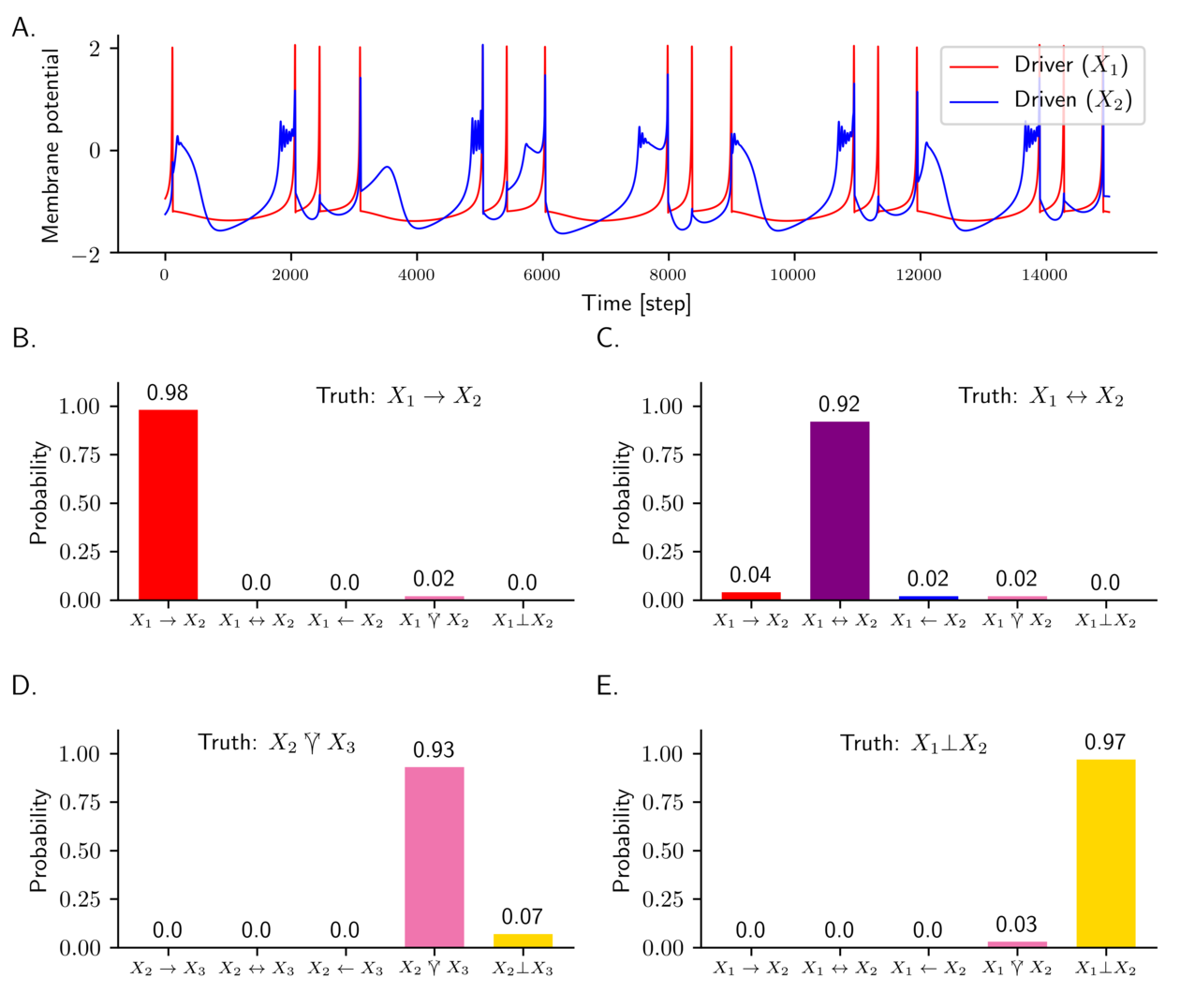}
	\caption{\textit{Testing our method on coupled Hindmarsh--Rose systems.}
    (A) The membrane potentials ($X_1 , X_2 $) of unidirectionally coupled Hindmarsh--Rose systems.
    (B)--(E) The inferred probabilities for the different types of couplings. (Color code matches with \autoref{fig:logmap} C.)
	}
	\label{fig:hindmarsh_rose}
\end{figure}

\autoref{fig:hindmarsh_rose} shows our results on different Hindmarsh--Rose systems. We conclude that in each case our method detects the proper causal relationship with high confidence. For further details about the Hindmarsh--Rose systems and model parameters we refer the reader to \ref{si:results}.

\subsection*{Changes of inter-hemispheric connectivity during photo-stimulation}

\begin{figure}
	\centering
	\includegraphics[scale=1.3]{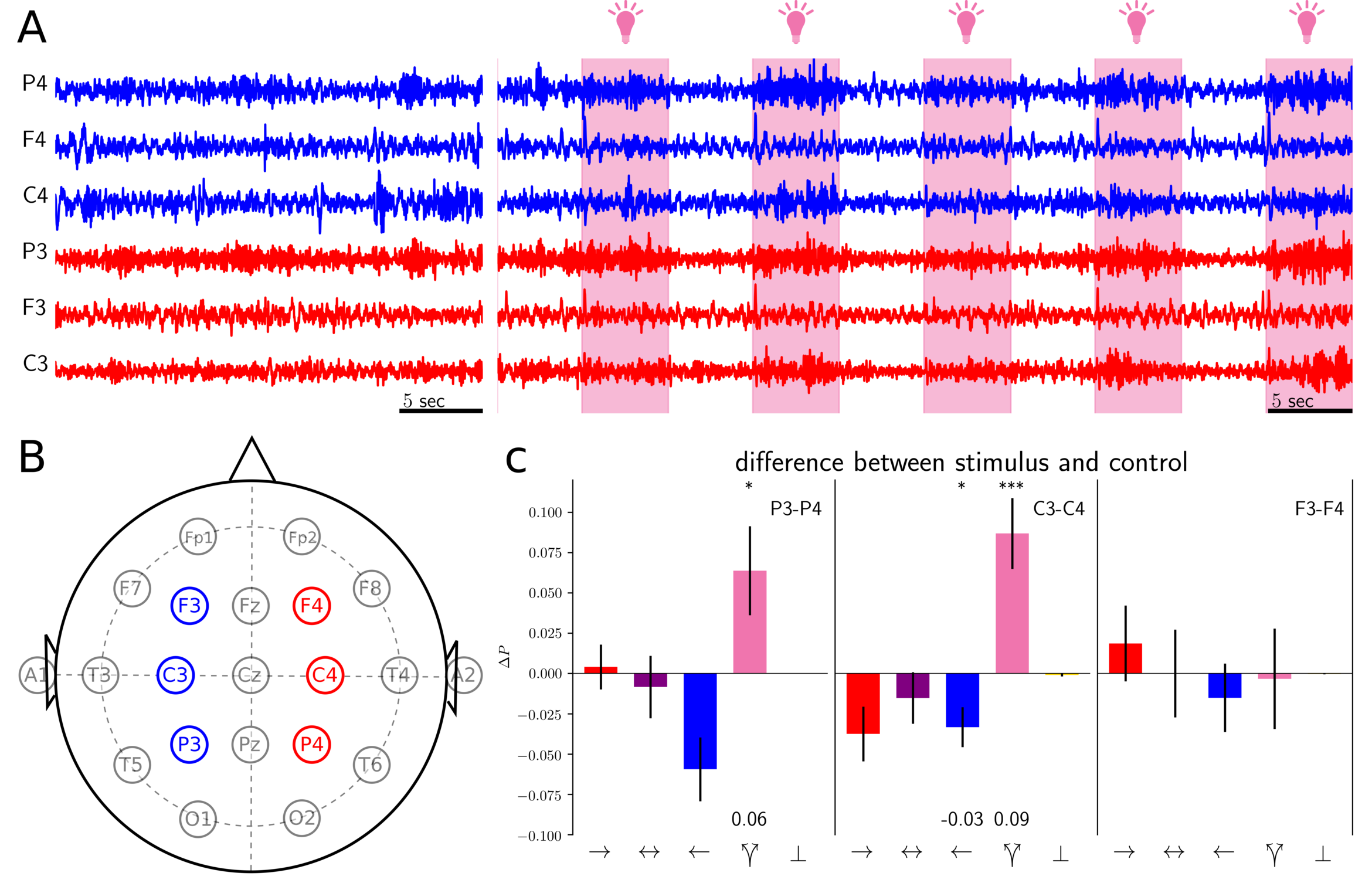}
	\caption{\textit{Inter-hemispherical interactions during  photo-stimulation.}
    (A) CSD signal in control condition and photo-stimulation periods (light bulbs) at the six analyzed recording-channels.
    (B) Electrode positions on the scalp. Causal relations were computed between $P3$--$P4$, $C3$--$C4$ and $F3$--$F4$ channel pairs.
    (C) Difference in probabilities of causal relations between stimulation and control (mean and SE). The probability of  the existence of common cause  is significantly higher during stimulation periods for $P3$--$P4$ ($p=0.024$) and $C3$--$C4$ ($p=0.0002$) channel-pairs but not for $F3$--$F4$ ($n=87$).}
	\label{fig:fotostim}
\end{figure}

Finally, we aim to assess our approach under real-world conditions, where the true dimensionality of the systems and the properties of the noise are unknown. In general, the precise causal relationships between time series in these systems are not known. However, external factors can induce changes in the internal causal relationships that our analysis method can detect. Notably, the standard epilepsy-diagnostic photo-stimulation procedure, where patients are exposed to flashing light at different frequencies in a standardized test, serves as an ideal model for an external common cause affecting the two brain hemispheres.

The connectivity between these brain regions suggests that visual information arrives first to ($O1$, $O2$ electrodes) and spreads to the parietal ($P3$--$P4$) and central ($C3$--$C4$) area (\autoref{fig:fotostim} B) afterwards. Hence, it is expected that our method is suitable to detect the visual stimulus as a common cause between the affected electrode sites.
Thus, we apply our method to Current Source Density (CSD, see \ref{si:results}) calculated from EEG recordings of 87 patients participating in the photo-stimulation task (\autoref{fig:fotostim} A).

The probability of common cause ($\hcc$) is significantly increased during visual stimulation periods for channel pairs $P3$--$P4$ and $C3$--$C4$ relative to the resting state but not for the frontal areas ($F3$--$F4$) (\autoref{fig:fotostim} C, \ref{si:results}).

\subsection*{Causal connections during epileptic seizure}

\begin{figure}
	\centering
    \includegraphics[scale=0.5]{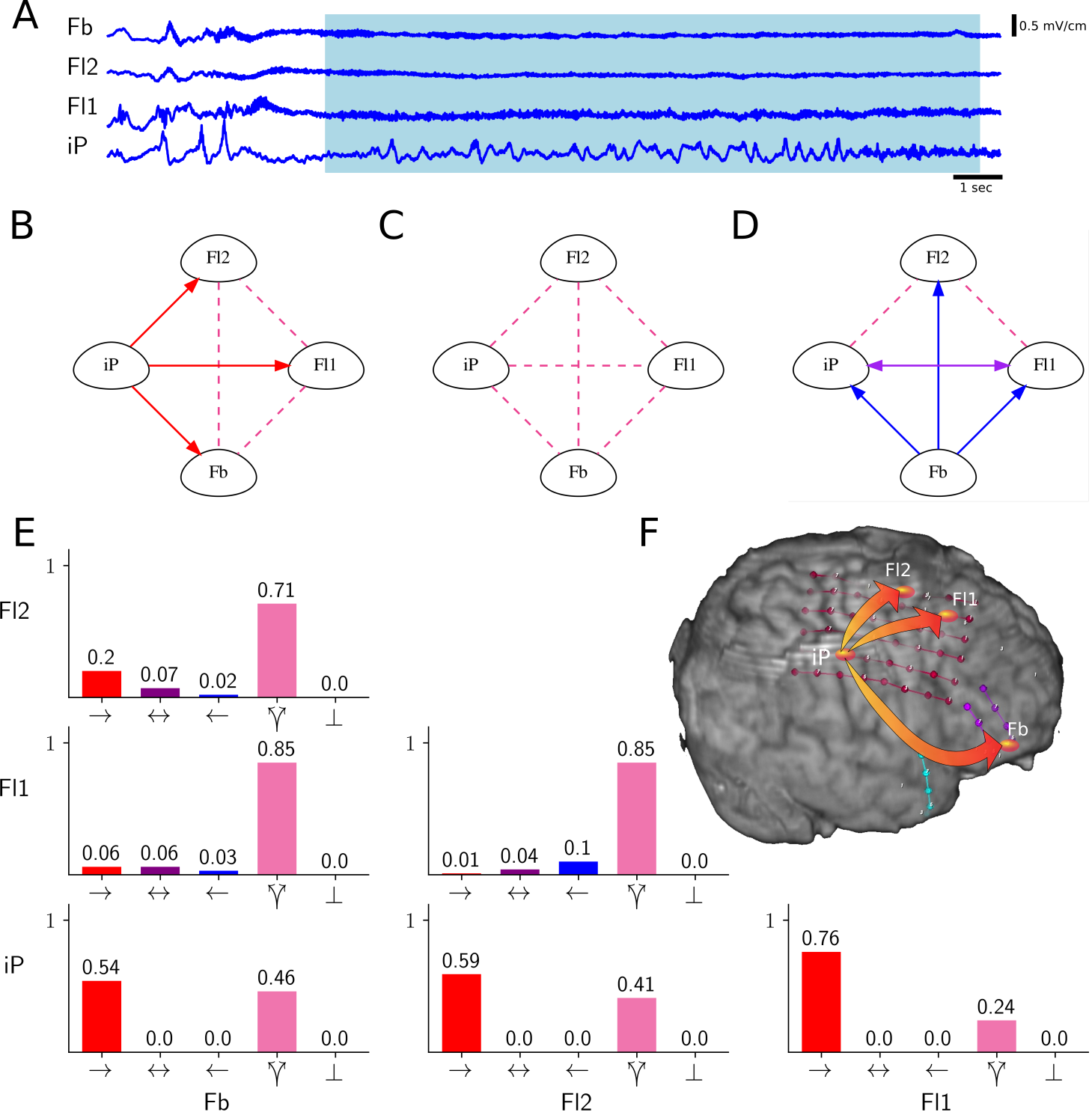}
    \caption{\textit{Cortical connectivity during epileptic seizure}
    (A) CSD signal at fronto-basal ($Fb$), frontal ($Fl1$) fronto-lateral ($Fl2$) and infero-parietal ($iP$) areas, the blue selection shows the analyzed time period of the seizure. Basically two types of connectivity were detected for seizures and a third type of connectivity for interictal conditions ($B$--$D$).
    (B) Maximal a posteriori probability (MAP) causal connection structure for the example seizure  in A. Red arrows mark unidirectional relations and the pink dashed lines mark detected common cause relations. (seizure type 1, $n=6$)
    (C) MAP causal connection structure for seizure type 2 ($n=10$)
    (D) MAP  mean causal connection structure interictal sections ($n=16$). Blue arrows mark unidirectional relations and purple arrows denote bidirectional causal relations.
    (E) Causal relation probabilities of the seizure showed on A and B.
    (F) The inferred driver and driven areas represented on the brain surface for the same seizure as on A, B, E.}
	\label{fig:memo}
\end{figure}

Our next example illustrates how identifying the causal connection between investigated areas can contribute to the diagnosis and surgical treatment of epileptic patients. Our investigations yields self-consistent results, affirming our assumption that our method is applicable in such complex situations.

We analyze EEG data from a 20-year-old patient suffering from drug-resistant epilepsy with frequent seizures. As part of the pre-surgical examination, a subdural grid and two strip electrodes were placed on the surface of the brain, enabling the identification of brain areas participating in seizure activity with high spatial precision.

The seizures presented a variable and complex pattern, with the majority of high-frequency seizure activity observed in the fronto-basal ($Fb$) and frontal ($Fl1$) regions (\autoref{fig:memo} C). The fronto-lateral ($Fl2$) region was moderately impacted by the seizures, and the infero-parietal ($iP$) region exhibited irregular high-amplitude spiking activity during both interictal and ictal periods (\autoref{fig:memo}~A).

We investigate all causal relations among the four brain regions and observe variable connectivity. Our causality analysis method reveals two main connectivity types during seizures (from the analysis of $n=18$ seizures, see Ext.\,Data\,\autoref{fig:memo_seizure_microfilm}, \ref{si:results}).
In the first main type, $iP$ is identified as the driver ($n=6$), while the pairwise analysis indicates the highest probability of a common cause between $Fb$, $Fl1$, and $Fl2$ (\autoref{fig:memo} B, E, F).
In the second main type of seizures, a hidden common cause is detected between all observed channels ($n=10$, \autoref{fig:memo} C), suggesting the possible existence of a driver beyond the investigated region. In contrast, $Fb$ is found to be the dominant driver node during normal interictal activity (\autoref{fig:memo} D, Ext.\,Data\,\autoref{fig:memo_control_microfilm}, \ref{si:results}).

Due to the difficult accessibility of the $iP$ region, the medical panel decided to resect the frontal and fronto-basal regions, leaving the less active areas intact. The patient remained seizure-free for 1 year before a relapse.

Our analysis underscores the variability of seizures and the potential existence of a hidden common cause behind observed seizure activity. This information could serve as a basis for understanding cortical reorganization and the recovery of epileptic activity in the investigated patient. However, a systematic analysis on a large population of patients will be necessary to clarify the underlying causal structures during epileptic network activity.

\section*{Conclusions}

We introduced a Bayesian causal inference method based on comparing the dimensionality of topologically embedded time series. It is the first unified approach for detecting all types of causal relations in dynamical systems. Specifically, it can uncover the existence of hidden common causes and quantifies the probability of each causal relationship type. We validated our method using synthetic data, and its self-consistency and practical applicability were demonstrated on human EEG data from epileptic patients.

The precision of our inference relies on the reliability of dimension estimations, influenced by several signal properties, including autocorrelation, sample size, noise level, and manifold geometry, especially curvature. These aspects are studied in a forthcoming paper. 
We have documented the disciplined methods we used to find well-behaving embedding dimension and other parameter settings (\ref{si:workflow}, also see \cite{Grassberger1983, Levina2005}).

The longer the signal's temporal autocorrelation, the more data points are needed to accurately sample the manifold's surface (in the reconstructed state-space) for inferring causal relations. Our simulations indicate that proper dimension determination typically requires at least a few thousand samples, which is more than what data Convergent Cross Mapping (CCM) needs to detect directed causal relations under optimal conditions. While operational with smaller sample sizes, CCM is less reliable in identifying hidden common causes and lacks assigned probabilities for causal cases. Despite being more data-intensive, our method robustly detects hidden common causes and assigns easily interpretable probabilities to each possible causal relation.

In general, the presence of observational or dynamic noise decreases the reliability of the inference, as demonstrated in \cite{hirata2016detecting}. Additionally, the curvature of the delay-reconstructed manifold causes the overestimation of local dimensionality (for a finite sample) and sets limits on the maximal noise level that allows meaningful state space reconstruction \cite{Casdagli1991}.

It is worth noting that correlation dimension and mutual information dimension have been applied to evaluate connections between dynamical systems, but neither provides an exhaustive description of bivariate causal relationships \cite{Hirata2010} due to their symmetric nature \cite{romano2016, Sugiyama2013} nor an inference machinery \cite{krakovska_2019}.

\subsection*{Future directions.} In the present paper, we focused on systems with deterministic dynamical components; however, we hope that our unified framework could be extended to pure stochastic systems \cite{Stippinger2023markov} and systems of mixed character \cite{Stippinger2023compression}, involving both deterministic and stochastic components.

In this work, we assumed no delay in the causal effect. Similar to \cite{ye2015distinguishing}, this assumption can be relaxed, allowing causality determination with different effect delays. Consequently, not only the existence of the connection but also the delay of the effect can be determined.

The pairwise determination of causal relationships used here could be expanded to multivariate comparisons to unveil more complex network structures \cite{Kurbucz2022linearlaw, Kurbucz2024NetDimension}.

In this paper, we gave a method to reveal the existence of a hidden common cause. After that we can take one step further to reconstruct it up to its topological equivalent
\cite{benko2022reconstructing}.

Our method serves as a unified approach to quantify all types of causal relationships. Based on the presented examples and tests, we believe that this new method will find applicability in various scientific areas. However, it is evident that there are many open questions and further directions to explore.
\bigskip

\begin{sloppypar}
{\bf Software availability:}
The implementation of the full inference method is freely available on GitHub at \href{https://github.com/adam-zlatniczki/dimensional_causality}{\nolinkurl{https://github.com/adam-zlatniczki}\allowbreak\nolinkurl{/dimensional_causality}}.
\end{sloppypar}
\bigskip

\begin{sloppypar}
{\bf Data availability:}
Raw photostimulation EEG dataset for \autoref{fig:fotostim} and Ext.\,Data\,\autoref{fig:fotostim_boxplot} is available at \url{https://web.gin.g-node.org/zsigmondbenko/photostim_eeg}.
Raw CSD signal of intra-cranial electrode recordings for \autoref{fig:memo}, Ext.\,Data\,\autoref{fig:memo_seizure_microfilm} and Ext.\,Data\,\autoref{fig:memo_control_microfilm}  are available at \url{https://web.gin.g-node.org/zsigmondbenko/intracranial_csd}.
\end{sloppypar}
\bigskip

{\bf Acknowledgement} The authors are grateful to dr.\ Boglárka Hajnal, dr.\ Ákos Újvári and dr.\ Anna Kelemen for their help during the clinical investigations and to Balázs Ujfalussy and Tamás Kiss for their comments on the manuscript. This research supported by grants from the Hungarian National Research, Development and Innovation Fund NKFIH K113147 (Z.S.), K135837 (Z.S.) and Human Brain Project associative grant CANON, under grant number NN 118902 (Z.S.), and the Hungarian National Brain Research Program KTIA NAP 13-1-2013-0001 (A.T., Z.S., D.F., E.L.) and KTIA-13-NAP-A-IV/1,2,3,4,6  (A.T., Z.S., D.F., E.L.). Hungarian Research Network, HUN-REN supported the project under the grant SA-114-2021 (Z.S.).
\bigskip

{\bf Statement} During the preparation of this work the authors used OpenAI's ChatGPT 3.5 in order to improve language style and grammar. After using this tool/service, the authors reviewed and edited the content as needed and take full responsibility for the content of the publication.
\bigskip

{\bf Author Contribution}
Z.B., A.Z., M.S., Z.S., A.T. worked out the causal inference method and wrote the manuscript. 
Z.B., A.Z. and M.S. ran the analysis on simulated and EEG data. 
Z.B., A.Z. and M.S. are equally contributed to the work.
A.Z. implemented the algorithm. 
A.S. recorded the photostimulation EEG data. 
L.E. and D.F. did the surgery and recorded intracranial LFP data. 
Z.S. and A.T. are equally contributed to the work, coined the original idea and developed the theory. 
All authors approved the final version of the manuscript.
\bigskip

{\bf Competing interests statement} The authors of the paper state that they do not have any competing interest.
\bigskip

{\bf Author Information} 
Correspondence and requests for materials should be addressed to \nolinkurl{somogyvari.zoltan@wigner.hun-ren.hu}.

\appendix

\section{Mathematical foundation} \label{Sapp}
\subsection{Intrinsic dimension and causality}\label{si:dim-cause}
In this section, we establish the mathematical background concerning the relationship between causal connections and the dimensions of the state space of systems.
We start with some general remarks on the possibilities and limitations of analyzing the connection of dynamical systems through an observation, a time series. It is standard to assume that the system is in a steady state, and consequently, the observed time series is stationary. We will examine the stationary distribution through its information dimension.
First, let us recall the definition of the information dimension introduced by Rényi \cite{Renyi1959}.

\begin{equation*}
d_{X}=\lim_{N\rightarrow \infty }\frac{1}{\log N} H([X]_N) ,
\end{equation*}
where $X\in \mathbb{R}^{m}$ is a  continuous  vector valued random variable, $[X]_N=\frac{\left\lfloor NX\right\rfloor }{N}$ is its $N$-quantization (a discrete variable) and $H(\cdot)$ stands for the Shannon entropy.
The finer the resolution $r=1/N$, the more accurate the approximation of the information content is. However, the normalization in the limit eliminates the contribution of the discrete part of the distribution (if it exists). This means that estimating the information dimension involves a trade-off and has limitations due to the finiteness of the sample. It is suggested that if the variable exists in a $D$-dimensional space, a proper estimate requires at least $10^D$ to $30^D$ sample points \cite{wolf1985}. Following Pincus' ideas \cite{pincus1991approximate, pincus1991regularity, pincus1995approximate}, we introduce the approximate information dimension.

\begin{equation*}
d_{X,1/N}=\frac{1}{\log N} H([X]_N) .
\end{equation*}
or in short $d_{X,r}$, for arbitrary partition with box size $r$.

\begin{assumption}
    The investigated time series are stationary.
\end{assumption}

\noindent{\bf Dimensions.}
The notion of dimension has many definitions depending on the context and methods. Takens' work was confined to topological dimensionality, but the embedding theorem was extended to fractal dimensions as well \cite{Sauer1991}. In our case, the attractor, the support of the stationary distribution, can be a fractal. We have already mentioned the Rényi information dimension.

On the other hand, the local intrinsic dimension is defined as follows.
Let $X \in \mathbb{R}^m $  be   the investigated random variable

\begin{equation*}
D_X(x)=\lim_{r\rightarrow 0 }\frac{1}{\log r} {\log P(X \in B(x,r) ) }
\end{equation*}
where $B(x,r)$ is the hypercube in $\mathbb{R}^m $ with lower corner $x$. Then the intrinsic dimension is $D_X=E(D_X(X))$
and	
\begin{equation*}
d_{X}=D_X
\end{equation*}
(see the work of Camastra and Staiano \cite{camastra2016intrinsic}, and theorems 1 and 2 of Romano et al.\ \cite{romano2016}).  ID estimation has a waste literature, several excellent reviews  \cite{campadelli2015intrinsic, camastra2016intrinsic} help to find the best performing one. Reviews and benchmark tests indicate that $k$-Nearest Neighbor-based methods have several advantages, including the rigorous derivation of the estimate presented by Levina and Bickel \cite{Levina2005} and the very good convergence properties of a variant \cite{Szepesvari2007, Benko2022dimension}. We follow the latter, particularly because it provides a hint for the choice of scaling, which is one of the crucial points in the intrinsic dimension (ID) estimation procedure.

\begin{assumption}
    The embedded manifolds  are homogeneous with respect to (the existing) dimension.
\end{assumption}
Based on the previous introduction our intrinsic dimension estimates are

\[
\widehat{D}\left( x\right)_r =\frac{1}{\log r}\log \left\vert N\left(x,r\right) \right\vert 
\]%
and%
\[
\widehat{D}_{X,r} =\frac{1}{n}\sum_{i=1}^{n}\widehat{D}\left(x_{i}\right)_r ,
\]%
where $n$ is the sample size, $\left\{ X_{i}\right\} _{i=1}^{n}$ is the set of  sample points on the  manifold and $N(x,r)= \left\lbrace X_i:X_i \in  B\left(x,r\right)  \right\rbrace$.

We shall use a bit reversed logic (following Farahmand, Szepesvári and Audibert \cite{Szepesvari2007}) and calculate the dimension from the distance of $k^\text{th}$ nearest neighbor
$r(x,k)=d(x,X^k(x) )$
where $X^k(x)$ is the $k^\text{th}$ closest point to $x$ in our sample series. In this setting the resolution is given by
\begin{equation}
r^D \approx \frac{k}{n}	
\end{equation}
where $n$ is the sample size.
\smallskip

\noindent{\bf Causal relations.}
Let us introduce the box partitioning of $\mathbb{R}^m$ using $r \mathbb{Z}^m$ and denote the boxes by $B(x,r)$, where $x$ is the lower left corner.  Also we index the partition element $B(x,r)$  by $x$.
The quantized variables are defined as follows.  
\begin{equation}
X^r=x \text{ if   } X\in B(x,r)
\end{equation}
\begin{definition}
    We say that $X$ causes/drives $Y$ at resolution $r$ (denoted by $X \rightarrow_r Y$)
    if there is a mapping $f$ s.t.\ $X_{t}^{r}= f \left( Y_{t}^{r}\right) $  for all $t$.
    \label{drive}
\end{definition}

In the case of a time-delayed causal relation between the systems with lag $\tau$, we have $X_{t-\tau}^{r} = f \left( Y_{t}^{r}\right)$; therefore, the proper time shift should be applied as a preprocessing step.

\begin{definition}
    We say that $X$ and $Y$ are in bidirectional causal relation at resolution $r$ (denoted by $X \leftrightarrow_r Y$), if both drives the other at resolution $r$.
\end{definition}

\begin{definition}
    We say that the $x_t$ and $y_t$ observations are basically identical (also called generalized synchrony) if there is a bijective mapping $f$ s.t.\ $x_{t}= f \left( y_{t}\right)$ for all $t$.
    From a dynamical systems point of view, generalized synchrony is a special case of bidirectional causality.
\end{definition}

One should note that we defined causality (Definition \ref{drive}) counterintuitively in a somewhat reversed direction. $X$ can be reconstructed from the information contained in $Y$, given that this is exactly what is conveyed from $X$ to $Y$ in the action of driving. If only partial reconstruction is possible, say $X=(X',Y'')$ and $Y=(Y',Y'')$, where all the components are independent, then $Y''$ is a common cause of $X$ and $Y$, and $X$ does not drive $Y$ or vice versa.

\begin{definition}
    We say that $X$ and $Y$ are independent at resolution $r$ (denoted by $X \perp_r Y$ if $X^r$ is independent of $Y^r$.  
\end{definition}
\begin{definition}
    We say that $\left\{ Y\right\} ,$ $\left\{ X \right\} $ have a common
    cause at resolution $r$ (denoted by $X$ $\mathrel{\scalebox{0.8}{$\curlyveeuparrow$}}_r$ $Y$)
    if they are, at resolution $r$, not independent and there is no driving connection between them again at resolution $r$.
\end{definition}

\noindent{\bf Link between causal and dimensional relations.}
First, we recall some elementary facts which will be useful. Let $\mathbb{J}=\left(X,Y\right)$ be the joint embedding. The intrinsic dimension $D$ coincides with the  information dimension $d$ and the same applies for the quantized versions \cite{romano2016}.

\begin{lemma}
    \label{LHH}
    The elementary properties of the discrete Shannon entropy imply that for all $r$
    
    \begin{equation}    
    \max \left\{ H\left( X^r \right) ,H\left( Y^r \right) \right\} \leq H\left(
    X^r,Y^r\right)  
    \leq H\left( X^r\right) +H\left( Y^r\right)
    \end{equation}

    and 
    \begin{equation}
    \max \left\{ D_{X,r},D_{Y,r}\right\} \leq D_{\left( X,Y\right),r }\leq D_{X,r}+D_{Y,r},
    \label{e2}
    \end{equation}
    where equality on the r.h.s.\ holds if and only if the variables are independent at resolution $r$ $\left(X \perp_r Y \right)$.   
\end{lemma}

\begin{corollary}
    \label{CDD}
    \label{bB}$X \perp_r Y$ iff.
    
    \begin{equation}
    D_{\left( X,Y\right),r }=D_{X,r}+D_{Y,r}.
    \end{equation}
    
\end{corollary}
\begin{proof}
    The statement is immediate from the additivity of the Shannon entropy of
    independent variables. 
\end{proof}    
\begin{corollary}
    If  $X\rightarrow_r Y$ then 
    
    \begin{equation*}
    D_{X,r}\leq D_{Y,r}=D_{\left( X,Y\right),r }
    \end{equation*}
    and if $X\leftrightarrow_r Y$ then
    
    \begin{equation*}
    D_{Y,r}=D_{X,r}=D_{\left( X,Y\right),r }.
    \end{equation*}
\end{corollary}

\begin{proof}
    Both statements follow from the fact that if $%
    X^r=f\left( Y^r \right) $ then 
    $H\left( X^r|Y^r\right)=H\left( f\left( Y^r\right)|Y^r \right)=0$:

    \begin{equation}    
    \begin{array}{rcl}
    H\left( X^r,Y^r\right) & = & H\left(X^r|Y^r \right) + H\left( Y^r\right) \\
    & = & H\left( f\left( Y^r\right)\right|Y^r) + H\left( Y^r\right) \\
    & = & H\left( Y^r\right).
    \label{ee3}
    \end{array}
    \end{equation}

\end{proof}

\begin{theorem}
    \label{Tdirect}If $D_{\left( X,Y\right),r }=D_{Y,r}$ then $X$ drives $Y$ at resolution $r$.
\end{theorem}

\begin{proof}
    It is immediate from the condition that
    
    \begin{equation*}
    \begin{array}{rcl}
    H\left(X^r|Y^r\right)=0
    \end{array}
    \end{equation*}
    but that implies that there is a mapping $f$ s.t.\ $X^r=f(Y^r)$.
    
\end{proof}

\begin{theorem}
    \label{cC} $X^r$ and $Y^r$ have a common cause if and only if
    
    \begin{equation*}
    \max \left\{ D_{X,r},D_{Y,r}\right\} < D_{\left( X,Y\right), r } < D_{X, r}+D_{Y, r}.
    \end{equation*}
\end{theorem}

\begin{proof}
    The statement follows from the combination of  Corollary \ref{CDD}  and
    Theorem~\ref{Tdirect}.
\end{proof}

If we combine all our results from Lemma \ref{LHH} to Theorem \ref{cC} then we have the following full table of implications.

\begin{equation*}
\begin{bmatrix}
X\rightarrow_r Y & \Longleftrightarrow  & D_{X,r}<D_{Y,r}=D_{\left( X,Y\right),r } \\ 
Y\rightarrow_r X & \Longleftrightarrow  & D_{Y,r}<D_{X,r}=D_{\left( X,Y\right),r } \\ 
X\leftrightarrow_r Y & \Longleftrightarrow & D_{X,r}=D_{Y,r}=D_{\left( X,Y\right),r }
\\ 
X\text{ }\mathrel{\scalebox{0.8}{$\curlyveeuparrow$}}_r\text{ }Y & 
\Longleftrightarrow  & \max \left\{ D_{X,r},D_{Y,r}\right\} <D_{\left( X,Y\right),r
}<D_{X,r}+D_{Y,r} \\ 
X\perp_r Y & \Longleftrightarrow & D_{\left( X,Y\right),r }=D_{X,r}+D_{Y,r}%
\end{bmatrix}
\end{equation*}

Let us note that the statements are confined to a finite resolution case.  If  $r \rightarrow 0 $ the statements remain valid if the distribution of the  driver system and $f$, the mapping function to the driven, is twice continuously differentiable.  So we assume, that the distribution of the information injected into the driven system, has no discrete part. 
\medskip

\noindent{\bf Causality detection using additive observation of the series.}
The genericity of Takens' theorem allows us to choose an observation function (mapping) almost freely. It is worth to note that the observation function shall be continuous to not introduce a discrete part in the manifold distributions.

As demonstrated in Eq.\ \eqref{e2}, the direct product of the embedding of the series has such good properties. On the other hand, it increases the embedding dimension, which is unfortunate in practice, so we may look for other mappings. The simplest, linear one, is a natural choice. Let us imagine $W=X+Y$. In general, it is again a good mapping, but if, in particular, $X=X_{1}-Y$, for independent $X_{1}$ and $Y$, the result would be a drop in dimension.

There is a simple resolution to that problem. If we choose a random number $a \in [0,1]$, where $a \neq 0$, uniformly, and consider $V=aX+Y$, then, with probability one, $V$ will be a proper observation function and can replace the direct product $J$. Also, if we know that $D_X$ and $D_Y$ are less than $m+1$, and we choose $a_i \in [-1,1], i=1,\ldots,m+1$ as random numbers, and set $V_i = a_iX-Y$, at least one should be a proper observation function, and the incorrect ones can be identified by the drop in dimension.

\subsection{Assigning probabilities to causal relations}\label{si:probabilities}
	
In the whole sequel, we omit marking the resolution $r$ where it does not cause confusion. Given the constraints $\max\left\lbrace{D_X,D_Y}\right\rbrace \leq D_J \leq D_Z \leq D_X+D_Y$, we have the next complete partition of the event space of possible causal relationships

\[
\begin{array}{cccc}
& \text{direct drive} &   \hcc &   \perp \\ 
D_{X}<D_{Y} & A_{1,1} & A_{1,2} & A_{1,3} \\ 
D_{X}>D_{Y} & A_{2,1} & A_{2,2} & A_{2,3} \\ 
D_{X}=D_{Y} & A_{3,1} & A_{3,2} & A_{3,3}%
\end{array}%
\]
and $\mathcal{A} = \left\lbrace A_{i,j}:i,j=1,2,3 \right\rbrace $ (in the main text, $A \in \mathcal{A}$ in the Bayesian argument). 
In the main text the causal relations referred as $X \rightarrow Y$, corresponds here to $A_{1,1}$, $X \leftarrow Y$ to $A_{2,1}$, $X \leftrightarrow Y$ to $A_{3,1}$, while $X \hcc Y = \sum_{i=1}^3 A_{i,2}$ and $X \perp Y = \sum_{i=1}^3 A_{i,3}$.

Based on the work of Romano et al.\ \cite{romano2016}, we consider the expected value of local dimensions to be the global dimension; therefore, the mean of the local dimension estimates yields our estimate of the global dimension. Consequently, the sampling distribution of the global dimension estimates is multivariate normal, with the true dimensions $D$ being its mean and covariance matrix $\Sigma$.

Let $\bar{D}$ denote the observed global dimension that depends on $D$ and $\Sigma$. Note, however, that $D$ and $\Sigma$ are uncertain parameters themselves since they depend on the causal model $A$. The dependence of $D$ on $A$ is trivial, as indicated by the links derived in the previous sections between causality and dimensions. However, $A$ also has an effect on the noise model $\Sigma$: for example, if the causal relationship is $X \rightarrow Y$ ($A = A_{1,1}$), then in theory, $D_Y = D_J$, and therefore, $D_Y$ and $D_J$ should be correlated.

Let $U$ represent the parameters of our method (e.g., $k$, $\tau$, embedding dimension $m$). Our method's parameters also affect dimension estimation, thereby influencing the noise model (but not the true dimensions $D$).

Finally, we can conclude that $\bar{D}$ depends on $D$ and $\Sigma$, $D$ depends only on $A$, while $\Sigma$ depends on both $A$ and $U$, where $U$ is the set of the hyperparameters.

In this setting the realization of $\bar{D}$ is a single 4-dimensional data point, let's denote it as $\bar{d}$. The likelihood of the data can be written as

\begin{equation*}
p_{\bar{D}}(\bar{d}) = \sum_{A_{ij} \in \mathcal{A}} p_{\bar{D}|A}(\bar{d}|A_{ij}) \; P_{A}(A_{ij}).
\end{equation*}

We are interested in the probabilities of each $A_{ij}$ given the data, therefore we apply Bayes' theorem and get

\begin{equation*}
P_{A|\bar{D}}(A_{ij}|\bar{d}) = \frac{p_{\bar{D}|A}(\bar{d}|A_{ij})}{p_{\bar{D}}(\bar{d})}P_{A}(A_{ij}).
\end{equation*}
\begin{sloppypar} Note that $p_{\bar{D}}(\bar{d})$ is only a normalizing term; therefore, it is enough to calculate $p_{\bar{D}|A}(\bar{d}|A_{ij})P_{A}(A_{ij})$ for all $A_{ij}$. We assume a non-informative prior over the possible causal relations, $P_A(A_{ij}) = \frac{1}{9}$ for all $i,j$. Taking the dependence structure of the random variables into consideration we can write the likelihood as \end{sloppypar}
\begin{equation*}
p_{\bar{D}|A}(\bar{d}|A_{ij}) = \int p_{\bar{D}|D,\Sigma}(\bar{d}|w,s) \; dP_{D,\Sigma|A}(w,s|A_{ij}).
\end{equation*}

If we extend this with $U$, i.e., the hyperparameters of the model, we get

\begin{equation*}
\begin{aligned}
p_{\bar{D}|A}(\bar{d}|A_{ij}) & = \int p_{\bar{D}|A,U}(\bar{d}|A_{ij}, u) \; dP_{U}(u) \\ 
& = \iint p_{\bar{D}|D,\Sigma}(\bar{d}|w,s) \; dP_{D,\Sigma|A,U}(w,s|A_{ij},u) \; dP_U(u).
\end{aligned}
\end{equation*}

We know that $\bar{d}$ comes from a 4-variate normal distribution with expected value vector $w$ and covariance matrix $s$, therefore

\begin{equation*}
\begin{aligned}
p_{\bar{D}|A}(\bar{d}|A_{ij}) & = \iint \varphi_{w,s}(\bar{d}) \; dP_{D,\Sigma|A,U}(w,s|A_{ij},u) \; dP_{U}(u) \\ 
& = \iint \varphi_{\bar{d},s}(w) \; dP_{D,\Sigma|A,U}(w,s|A_{ij},u) \; dP_{U}(u) \\ 
& = \iiint \varphi_{\bar{d},s}(w) \; dP_{D|A,U}(w|A_{ij},u) \; dP_{\Sigma|A,U}(s|A_{ij},u) \; dP_{U}(u) \\ 
& = \iiint \varphi_{\bar{d},s}(w) \; dP_{D|A}(w|A_{ij}) \; dP_{\Sigma|A,U}(s|A_{ij},u) \; dP_{U}(u),
\end{aligned}%
\end{equation*}
where $\varphi_{w,s}$ denotes the multivariate normal PDF with expected value vector $w$ and covariance matrix $s$. We used the conditional independence of $D$ and $\Sigma$, the independence of $D$ and $U$, and that in the probability density function of the normal distribution the expected value and the data can be exchanged, since $(\bar{d} - w)^T s^{-1} (\bar{d} - w) = (w - \bar{d})^T s^{-1} (w - \bar{d})$.

Let us now consider $dP_{D|A}(w|A_{ij})$. Every $A_{ij}$ induces a set $S_{ij} \subset \mathbb{R}_+^4$ such that each element of $S_{ij}$ satisfies the conditions given by $A_{ij}$. For example, $S_{1,1} = \{v : v \in \mathbb{R}_{+}^{4}, v_1 < v_2 = v_3 < v_4\}$. It is easy to verify that $S_{ij}$ is a convex cone with (algebraic) dimension $q_{ij} \in \{2, 3, 4\}$, simply embedded into a $4$-dimensional Euclidean space. Let $C_{ij} \in \mathbb{R}_+^{q_{ij}}$ denote the convex cone in its original, lower dimensional space. For example $C_{1,1} = \{w : w \in \mathbb{R}_{+}^{3}, w_1 < w_2 < w_3\}$. For any $A_{ij}$ there exists a simple linear mapping $M_{ij}$ such that $M_{ij}w \in S_{ij} \; \forall w \in C_{ij}$, for example
\begin{equation*}
    M_{1,1} = \left[\begin{matrix}
    1 & 0 & 0 \\
    0 & 1 & 0 \\
    0 & 1 & 0 \\
    0 & 0 & 1
    \end{matrix}\right]
\end{equation*}
which provides the correspondence between $S_{ij}$ and $C_{ij}$.

We assume a non-informative uniform prior on the true dimensions $D$. Since the prior is conditioned on $A_{ij}$, $C_{ij}$ becomes its support, resulting in an improper prior (due to the cone being infinite). It is obvious that this is not an actual distribution, but specifies a prior with correct proportions (granting equal weight to each element of the support). Putting these together we get

\begin{equation*}
p_{\bar{D}|A}(\bar{d}|A_{ij}) = \displaystyle\iiint_{C_{ij}} \varphi_{\bar{d},s}(M_{ij}w) \; dw \; dP_{\Sigma|A,U}(s|A_{ij},u) \; dP_{U}(u).
\end{equation*}

Let us now focus on marginalizing $\Sigma$. The $4 \times 4$ covariance matrix has 10 elements in the upper-triangular part. Marginalizing each of these elements presents computational challenges and a lack of knowledge about the conditional distribution of $\Sigma$, both of which we aim to avoid. Instead, we calculate a sample covariance matrix $\hat{s}$, which is a maximum-likelihood estimate, and assume that the distribution is highly peaked at this estimate.

Note, however, that the neighborhoods of our sample points overlap, and the calculated local dimensions are not independent. 
Therefore, the covariance matrix must be calculated by taking the correlation of local dimension estimates into consideration. If the samples were independently drawn, then the covariance of the means would be given by the covariance of the local dimensions divided by the number of samples. If the samples are correlated, we must divide by the effective sample size instead. Loosely speaking, the effective sample size of an estimator of the population mean is the number with the property that our estimator has the same variance as the estimator achieved by sampling the same amount of independent individuals. In our case, two local dimension estimates are independent if their $k$-neighborhoods do not intersect. Therefore, we can (approximately) sample $\frac{n}{2k}$ independent elements from them, and that is our effective sample size.

While a specific $A_{ij}$ explicitly excludes certain dimension combinations, it cannot exclude a covariance matrix. $\hat{s}$ has a positive likelihood given any causal relationship. We assume that the distribution of $\Sigma$ is highly peaked at the maximum-likelihood estimate, to the extent that we regard it as a Dirac delta, $dP_{\Sigma|A,U}(s|A_{ij},u)=dH(s - \hat{s})$, where $H(\cdot)$ is a multivariate unit step function. Note that $\hat{s}$ is a function of $U$ as well.

After incorporating this into the integral and marginalizing we get

\begin{equation*}
p_{\bar{D}|A}(\bar{d}|A_{ij}) = \iint_{C_{ij}} \varphi_{\bar{d},\hat{s}}(M_{ij}w) \; dw \; dP_{U}(u).
\end{equation*}
In the $A_{ij} = A_{1,1} = \{X < Y = J < Z\}$ case the above would result in

\begin{equation*}
p_{\bar{D}|A}(\bar{d}|A_{1,1}) = \displaystyle \int	\left(\int_{0}^{\infty}\int_{v_1}^{\infty}\int_{v_2}^{\infty} \varphi_{\bar{d},\hat{s}}(w_1, w_2, w_2, w_3) dw_3 \, dw_2 \, dw_1\right) dP_{U}(u).
\end{equation*}

In the current implementation $U$ consists only of $k$ (the neighborhood size, which is discrete uniform), therefore integrating by $U$ practically boils down to averaging. The other $A_{i,j} \in \mathcal{A} $ cases can be treated similarly.

\section{Methods}
\subsection{Analysis workflow}\label{si:workflow}

Our proposed causality analysis method begins with two time series. First, both time series must undergo data cleaning: one needs to ensure that the data is stationary (which can be checked, for example, with an augmented Dickey-Fuller unit root test) and that observational noise is addressed (for example, with a filter). Transforming the series into stationary ones is non-trivial and typically depends on the scientific field, requiring careful attention and field expertise. Choosing the wrong transformation may result in the loss of valuable information or the introduction of artifacts, biasing the final results. For example, differentiating the time series may remove too much information, using moving averages to smooth the series introduces higher autocorrelation, or calculating current source density with non-disjoint sets of signals to remove correlation (which only indicates a first-order relationship) may introduce higher-order dependence.
Another important preprocessing phase is normalization: if the two time series have different scales (or magnitudes), the results of dimension estimation will become biased. For example, if values of $X$ are around $1$ while values of $Y$ are around $100$, then the dimension estimation of $X$ and $Y$ will be correct. However, since $J = (X,Y)$ or $J = aX + Y$, the k-NN distances will be the same in $J$ and $Y$, dominating the effects of $X$. There are several ways to normalize the data, such as $[0,1]$-scaling, z-scores, quantile normalization (or rank normalization), etc.

The preprocessed time series are embedded into an $m$-dimensional space with lag $\tau$ according to Takens' theorem. Both $m$ and $\tau$ are parameters of our model that need to be specified. Takens showed that $m = 2d + 1$ (where $d$ is the true dimension) is a sufficient choice for the embedding dimension to reconstruct dynamics. Unfortunately, the true dimension is usually unknown, and one may need to consult a field expert with sufficient knowledge about the dynamical system to provide a proper estimation of the true dimension. If this is not possible, one can use several methods proposed in the literature for the selection of $m$. The correct value of $m$ can be determined by searching for a plateau (saturation effect) in invariant quantities (such as intrinsic dimensionality) or by using the false nearest neighbor method \cite{Cao1997, Rhodes1997}.
We applied an iterative process, starting with a high embedding dimension and decreasing it, checking the estimated manifold dimensions after each decrease, and selecting the lowest possible $m$ that did not sharply reduce the estimated $d$ value. While Takens suggests that $m = 2d + 1$ is a sufficient choice, it does not mean that smaller $m < 2d+1$ cannot be acceptable; this depends on the system. In general, \cite{Casdagli1991} states that self-intersections do not alter the estimated manifold dimension if $m > d$. In finite samples, it is common for dimension estimates to increase as the embedding dimension increases; therefore, in general, it is better to select $m$ as low as possible.

The optimal value of $\tau$ can be determined from the first zero point of the autocorrelation function or from the first minima of the automutual-information function \cite{Fraser1986}. Additionally, one can optimize for $m$ and $\tau$ simultaneously by applying differential geometric arguments or using the statistical properties of dimension estimates calculated on embedded data \cite{Nichkawde2013, Tamma2016}, or determine an optimal embedding window \cite{Small2004} $(m-1)\tau$. To find the optimal $\tau$, we analyzed the partial autocorrelation function (PACF) of the time series and selected the largest significant lag (which is how one would fit an autoregressive model).
The first insignificant lag in the autocorrelation function (ACF) could be used as well, but we found PACF more effective in practice. The largest significant lag of PACF as an embedding delay results in relatively independent (moderately redundant) coordinates but still not too independent (irrelevant) to reconstruct dynamics. The ACF may diminish very slowly, resulting in a very large $\tau$, or, for example, if it is monotonic and we select the first insignificant lag, we exclude the linear and their induced non-linear relationships.
On the other hand, the largest significant lag in PACF tells us which is the largest lag where there is no direct linear relationship, but non-linear relationships induced by these linear ones remain (think of the classical partial correlation example when there is a direct linear relationship between $t_1$ and $t_2$, $t_2$ and $t_3$, and therefore a quadratic between $t_1$ and $t_3$).

Given the two $m$-dimensional embedded manifolds of the series, $X$ and $Y$, the joint $J$ and independent $Z$ manifolds are also created. For this purpose, one can either choose direct products or an additive observation function ($aX + Y$). We advise the usage of an additive observation function since, as we have shown, the two result in equivalent relationships between manifold dimensions. However, direct products produce higher-dimensional product spaces where dimension estimation becomes more unreliable. Our experience showed that $a = \sqrt{\frac{29}{31}}$ is a good choice in general because it ensures that the scale of the series does not change too much.

In some cases, the embedded manifolds may have to undergo further transformations. For example, continuous dynamics (like a Lorenz system) evolve rather slowly and in a thread-like manner. If one does not choose a sufficiently large neighborhood when estimating the local dimension at a point, the manifold may seem one-dimensional because the nearest points will all lie on the same thread. This effect is similar to sampling a slow process with too high frequency. This may be handled by down-sampling the manifold, which means re-sampling it with a lower frequency.

The (post-processed) manifolds are used to estimate local dimensions, for which we employ an estimator proposed by Farahmand, Szepesvári, and Audibert \cite{Szepesvari2007}. Their estimator has one parameter, $k$, representing the size of the neighborhood around a point in which the dimension will be estimated. Depending on the system, this estimator may be sensitive to this parameter; therefore, we try several values and aggregate the results, as described in the previous section.

As stated before, dimension estimation is challenging, and local estimates may deviate significantly from the actual dimension. To handle this, one should remove outliers from the estimates. We recommend trimming the dimension estimates by dropping those that belong to the upper or lower $\alpha$-percent tail.

At this point, all the data is available for estimating model probabilities, as described in the previous section. We repeat this process for a range of $k$ values, and assuming a uniform prior over them, we obtain the final model probabilities by practically averaging the probability of each model for each $k$.

\subsection{Empirical results}\label{si:results}

\noindent{\bf Logistic maps.}
For testing purposes, we applied our method to systems of three-coupled logistic maps with various connectivity patterns. The logistic map is a simple non-linear discrete-time dynamical system that serves as a model for various economic and ecosystems, capable of producing chaotic behavior even in one dimension. It is defined as

\begin{equation}
x_j[t+1] = r x_j[t] (1 - \sum \beta_{jl} x_l[t]),
\label{eq:logmapNonLin}
\end{equation} 
where $r=3.99$, $j, l \in \lbrace 1, 2, 3 \rbrace$ are indices for the three variables and $\beta_{jl}$ are the elements of the coupling matrix ($\underline{\underline{B}}$) according to an actual coupling scenario:
\begin{enumerate}[(a)]
    \item Direct coupling:
    \[ \underline{\underline{B}} =
    \begin{bmatrix}
    1	&	0 &		0  \\
    \beta_{21}	&	1 &		0  \\
    0	&	0 &		1  \\
    \end{bmatrix}
    \]
    
    \item Bidirectional coupling:
    \[ \underline{\underline{B}} =
    \begin{bmatrix}
    1	&	\beta_{12} &		0  \\
    \beta_{21}	&	1 &		0  \\
    0	&	0 &		1  \\
    \end{bmatrix}
    \]
    
    \item Common cause case:
    \[ \underline{\underline{B}} =
    \begin{bmatrix}
    1	&	0 &		\beta_{13}  \\
    0	&	1 &		\beta_{23}  \\
    0	&	0 &		1  \\
    \end{bmatrix}
    \]
    
    \item Independent case:
    \[ \underline{\underline{B}} =
    \begin{bmatrix}
    1	&	0 &		0  \\
    0	&	1 &		0  \\
    0	&	0 &		1  \\
    \end{bmatrix}
    \]
\end{enumerate}
where $\beta_{12}=\beta_{21}=\beta_{13}=\beta_{23}=0.5$ in the example shown in the main text.

We simulated unidirectional, bidirectionally causal, independent and hidden common cause connection patterns (Fig.\,1) with $N=10\,000$ time series length. Only the first two subsystems were observed, the activity of the third subsystem was hidden.

We preprocessed time series data by applying rank normalization.

We set model-parameters and applied our method on logistic map datasets. We set embedding delay to $\tau=1$ and we found that embedding dimension $m=4$ was big enough in all cases. The probabilities were averaged over the neighborhood sizes $k=[12,44]$, where the dimension-estimates were constant.

Our method was able to reconstruct the original coupling pattern between the observed logistic maps for all test cases; curiously, it was able to detect hidden common cause between the two observed logistic maps.

In order to test the effect of a hidden common drive to Sugihara's CCM method, two types of coupling were applied: The case of non-linear coupling corresponds to equation (\ref{eq:logmapNonLin}), however the form for linear coupling is slightly different:

\begin{equation}
x_j[t+1] = r x_j[t] (1 - x_j[t]) +\sum \beta_{jl} x_l[t]
\label{eq:logmapLin}
\end{equation}

Sugihara et al.\ \cite{Sugihara2012} stated that high values of CCM, which are independent of the data length, are a sign of a hidden common cause. High CCM for short data series is typical for highly correlated data series. Similarly, Harnack et al.\ \cite{harnack2017topological} apply the existence of the correlation without detectable direct causality to reveal a hidden common cause. Our simulation results show that the above-described properties of CCM hold only for those cases in which the hidden common cause is linearly coupled to the observed variables, implying higher linear correlation ($r=0.46$) among them (Extended Data \autoref{fig:CommonCauseTest}, blue and green lines). If the hidden common driver is non-linearly coupled to the observed time series, thus the implied correlation is low ($r=0.15$), the CCM increases with the increasing data length but still remains low (Extended Data \autoref{fig:CommonCauseTest}, black and red lines). We conclude that the data length dependence of the CCM or the presence of correlation without causality does not allow us to reliably reveal or distinguish the existence of a hidden common cause; the CCM presumably follows only the linear correlation between the two driven variables in this case. Note that our method was tested and works well on the non-linearly coupled case as well.
\medskip

\noindent{\bf Lorenz systems.}
There are several ways to couple three Lorenz-systems; we implement the coupling through their $X$ coordinate such that

\begin{equation*}
\begin{array}{l}
\dot{X_1} = \sigma (Y_1 - X_1) + c_{2 \rightarrow 1}(Y_1 - X_2) + c_{3 \rightarrow 1}(Y_1 - X_3), \\
\dot{Y_1} = X_1 (\rho - Z_1) -Y_1, \\
\dot{Z_1} = X_1 Y_1 - \beta Z_1, \\
\\
\dot{X_2} = \sigma (Y_2 - X_2) + c_{1 \rightarrow 2}(Y_2 - X_1) + c_{3 \rightarrow 2}(Y_2 - X_3), \\
\dot{Y_2} = X_2 (\rho - Z_2) - Y_2, \\
\dot{Z_2} = X_2 Y_2 - \beta Z_2, \\
\\
\dot{X_3} = \sigma (Y_3 - X_3), \\
\dot{Y_3} = X_3 (\rho - Z_3) - Y_3, \\
\dot{Z_3} = X_3 Y_3 - \beta Z_3,    
\end{array}
\end{equation*}
where $(X_i, Y_i, Z_i)$ are the three coordinates of the $i^\text{th}$ system, $\sigma, \rho$ and $\beta$ are model parameters, and $c_{i \rightarrow j}$ denotes the strength of coupling from the $i^\text{th}$ system to the $j^\text{th}$ system.

For the simulations we set $\sigma = 10$, $\rho = 28$, $\beta = 8/3$, $\Delta t = 0.01$. The initial conditions are $X_1(0) = 10$, $Y_1(0) = 15$, $Z_1(0) = 21.1$, $X_2(0) = 17$, $Y_2(0) = 12$, $Z_2(0) = 14.2$, $X_3(0) = 3$, $Y_3(0) = 8$, $Z_3(0) = 12.4$, and we take $200\,000$ samples. The coupling coefficients are chosen for the four cases as follows:

\begin{enumerate}[(a)]
    \item Direct cause: $c_{1 \rightarrow 2} = 3.5$, $c_{2 \rightarrow 1} = c_{3 \rightarrow 1} = c_{3 \rightarrow 2} = 0$,
    \item Bidirectional cause: $c_{1 \rightarrow 2} = c_{2 \rightarrow 1} = 3.5$, $c_{3 \rightarrow 1} = c_{3 \rightarrow 2} = 0$,
    \item Hidden common cause: $c_{1 \rightarrow 2} = c_{2 \rightarrow 1} = 0$, $c_{3 \rightarrow 1} = c_{3 \rightarrow 2} = 3.5$,
    \item Independence: $c_{1 \rightarrow 2} = c_{2 \rightarrow 1} = c_{3 \rightarrow 1} = c_{3 \rightarrow 2} = 0$.
\end{enumerate}

Our method requires the specification of the following parameters: the embedding dimension $m$, time-delay $\tau$ and a set of integers for different $k$-s for the $k$-NN search. Takens found that if a manifold has dimension $D$ then $2D + 1$ is a good candidate for the embedding dimension. The dimension of the joint manifold is at most $6$ (the number of state variables). We try to minimize the embedding dimension: we start with $2D+1$ and decrease it as long as the dimension estimates are not forced down. This way we found the following embedding dimensions for the four causal cases: (a) 7, (b) 7, (c) 7 and (d) 5.
We select $\tau = 25$ based on the partial autocorrelation function (PACF) of the time series and $k \in [10, 38]$. Due to the threadlike phase space we also apply downsampling on the reconstructed manifold by keeping only every fourth point. We drop the outlying dimension estimates that belong to the lower or upper 5\% tail.
\medskip

\noindent{\bf Hindmarsh-Rose systems.}
We use three electrically coupled Hindmarsh-Rose neurons where coupling is achieved through their membrane potential ($X$), described by the differential equations

\begin{equation*}
\def\arraystretch{1.5}
\begin{array}{l}
\dot{X_1} = Y_1 - a X_1^3 + b X_1^2 - Z_1 + I_1 + c_{2 \rightarrow 1}(X_2 - X_1), \\
\dot{Y_1} = c - d X_1^2 - Y_1, \\
\dot{Z_1} = r_1\left(s\left(X_1 - \chi\right) - Z_1\right), \\
\\
\dot{X_2} = Y_2 - a X_2^3 + b X_2^2 - Z_2 + I_2 + c_{1 \rightarrow 2}(X_1 - X_2), \\
\dot{Y_2} = c - d X_2^2 - Y_2, \\
\dot{Z_2} = r_2\left(s\left(X_2 - \chi\right) - Z_2\right), \\
\\
\dot{X_3} = Y_3 - a X_3^3 + b X_3^2 - Z_3 + I_3 + c_{1 \rightarrow 3}(X_1 - X_3), \\
\dot{Y_3} = c - d X_3^2 - Y_3, \\
\dot{Z_3} = r_3\left(s\left(X_3 - \chi\right) - Z_3\right),
\end{array}
\end{equation*}
where $a, b, c, d, \chi, r_j, s,$ and $I_j$ are model parameters and $c_{i \rightarrow j}$ denotes the strength of coupling from the $i^\text{th}$ neuron to the $j^\text{th}$.
Note that the first neuron is coupled into the remaining two, while the second neuron is coupled into the first one, meaning that the second and third neurons are not coupled directly, only through a hidden common driver. Common choices for the parameters are $a = 1$, $b = 3$, $c = 1$, $d = 5$, $s = 4$, $\chi = -1.6$. It is also important that $r_i$ should be in the magnitude of $10^{-3}$ and $I_i \in [-10, 10]$; we choose $r_1 = 0.001$, $r_2 = r_3 = 0.004$, $I_1 = 2.0$, $I_2 = 2.7$ and $I_3 = 2.4$. The initial conditions are $X_1(0) = 0$, $Y_1(0) = 0$, $Z_1(0) = 0$, $X_2(0) = -0.3$, $Y_2(0) = -0.3$, $Z_2(0) = -0.3$, $X_3(0) = 0.3$, $Y_3(0) = 0.3$ and $Z_3(0) = 0.3$. We take $70\,000$ samples with time steps $\Delta t = 0.1$ and drop the first $10\,000$ where the system is still in a burn-in period. The coupling coefficients are chosen for the four cases as follows:

\begin{enumerate}[(a)]
    \item Direct cause: $c_{1 \rightarrow 2} = 0.615$, $c_{2 \rightarrow 1} = c_{1 \rightarrow 3} = 0$,
    \item Bidirectional cause: $c_{1 \rightarrow 2} = c_{2 \rightarrow 1} = 0.615$, $c_{1 \rightarrow 3} = 0$,
    \item Hidden common cause: $c_{1 \rightarrow 2} = c_{1 \rightarrow 3} = 0.615$, $c_{2 \rightarrow 1} = 0$,
    \item Independence: $c_{1 \rightarrow 2} = c_{2 \rightarrow 1} = 	c_{1 \rightarrow 3} = 0$.
\end{enumerate}
In each case we use the observations of the first two neurons, except for the hidden common cause case where we use the second and third neurons.

We parametrize the method based on the ideas pointed out in the Lorenz case. We set $k \in [10, 98]$, select $\tau = 5$ based on the PACF of the time-series, and the embedding dimension to be (a) 4, (b) 5, (c) 5, and (d) 3. We drop the lower and upper 5\% of the dimension estimates.
\medskip

\noindent{\bf Changes of inter-hemispheric connectivity during photo-stimulation.}

\noindent{\textit{\textbf{EEG recordings and photo-stimulation.}}} Monopolar recordings were taken with  standard electrode arrangement of the 10--20 system by a Brain Vision LLC (Morrisville, NC 27560, USA) EEG device with $500$ Hz sampling frequency from $n=87$ patients and the data was stored by the Vision Recorder software with a $0.1$--$1000\,\mathrm{Hz}$ bandpass filter enabled (no notch filter).
Flashing light stimulation were carried out 
in a $0.1$--$29\,\mathrm{Hz}$ range.

\noindent{\textit{\textbf{Preprocessing.}}}
We computed 2D current source density (CSD) of the EEG signal at the six selected electrode positions ($P3, P4, C3, C4, F3, F4$) to filter out linear mixing between channels \cite{Trongnetrpunya2016}. CSD is roughly proportional to the negative second derivative of the electric potential, so we can represent CSD by the  discrete Laplace of the original EEG signal up to a constant factor by the simplified formula:
	
\begin{equation}
CSD(x) \propto \sum_{i=1}^N n_{i}^{(x)} - N  x
\label{equation_csd} 
\end{equation}
where $x$ is the EEG signal measured at an electrode position, $n_{i}^{(x)}$ is the signal measured at a neighboring contact point, $N$ is the number of neighbors and $CSD(x)$ is the CSD at $x$. 
We have to mention that Eq.\ \eqref{equation_csd}  overlooks the non-uniform spatial distances of neighbors.
The following neighborhoods were used for the EEG channels:
\begin{itemize}[--]
    \item P3: \{C3, T5, O1\},
    \item P4: \{C4, T6, O2\},
    \item C3: \{T3, F3, P3\},
    \item C4: \{T4, F4, P4\},
    \item F3: \{Fp1, F7, C3\},
    \item F4: \{Fp2, F8, C4\}.
\end{itemize}

Medial electrode points (common neighbors) were omitted from the neighborhoods to avoid the artificial introduction of common cause relations.

Moreover $1$--$30\,\mathrm{Hz}$ band-pass filtering ($4^\text{th}$ order Butterworth filter) was applied to the CSD signal. 
From stimulation periods, we used only the  $21$, $24.5$, $26$, $29\,\mathrm{Hz}$) stimulation segments in the analysis, which induced the greatest evoked EEG activity on recording channels.
The resulting time series segments were $46\,760$--$84\,907$ samples long for stimulations.
Artifactless segments with matching sample sizes were taken as paired control from stimulation-free periods.
	
\noindent{\textit{\textbf{Model parameters.}}}
We determined model parameters embedding delay and embedding dimension empirically, as pointed out, e.g., in the Lorenz case.
Embedding delay was set according to the first insignificant value of PACF on the first dataset, and also a rough screening was made where we applied the method with unrealistically big delays and compared the results with smaller delay values. We assumed that choosing unrealistically big $\tau$ yields erroneous results, and a smaller $\tau$ which yields different result may be a valid choice of the embedding parameter. Both approaches showed that $\tau=5$ was a reasonable value for embedding delay.

Embedding dimension was set according to a screening from $D=5$ to $D=14$, with $D=7$ as the smallest dimension required for the embedding. The joint manifold ($J'$) was created with $a=1$.
\medskip

\noindent{\bf Causal connections during epileptic seizure.}

\noindent{\textit{\textbf{Electrode implantation and recording.}}}
Patients had medically intractable seizures and were referred for epilepsy surgical evaluation. 
After detailed non invasive EEG, video-EEG and various imaging studies, including MRI and PET, they underwent phase 2 invasive video-EEG examination. 
Intracranial subdural electrodes (grids and strips) were utilized to localize the epileptogenic zone.
In the present study, we used the original data obtained along the clinical investigation without any influence by the study.
All the patients provided informed consent for clinical investigation and surgeries along institutional review board guidelines, in accordance with the Declaration of Helsinki.

The patient underwent subdural strip and grid electrode implantation (AD-TECH Medical Instrument Corp., Racine, WI).
Subdural electrodes ($10\,\mathrm{mm}$ intercontact spacing) were implanted with the aid of neuronavigation and fluoroscopy to maximize accuracy \cite{Eross2009}, via craniotomy with targets defined by clinical grounds.
Video-EEG monitoring was performed using Brain-Quick System Evolution (Micromed, Mogliano Veneto, Italy).
All signals were recorded with reference to the skull or mastoid at $1024\,\mathrm{Hz}$ sampling rate.
To identify the electrode locations, the patient received an anatomical $T1$-weighted MPRAGE MRI before electrode implantation as well as a full head CT scan ($1\,\mathrm{mm}$ slices) after electrode implantation. 
Electrode locations were identified on the postimplantation CT scan using the software BioImage Suite \cite{Duncan2004} (\url{www.bioimagesuite.org}).
These locations were then mapped to the pre-implant MRI via an affine transformation derived from coregistering the pre-implant and post-implant MRIs and post-implant MRI and CT scans using FLIRT \cite{Jenkinson2001} and the skull-stripping BET 2 algorithm \cite{Smith2002}, both part of the Oxford Centre for Functional MRI of the Brain (FMRIB) software library (FSL: \url{www.fmrib.ox.ac.uk/fsl}).
The reconstructed pial surface was computed from the pre-implantation MRI using FreeSurfer \cite{Dale1999} (\href{https://surfer.nmr.mgh.harvard.edu}{\nolinkurl{surfer.nmr.mgh.harvard.edu}}) and the electrode coordinates projected to the pial surface \cite{Dykstra2012} to correct for possible brain shift caused by electrode implantation and surgery.
Intraoperative photographs and ESM were used to corroborate this registration method.
This pial surface projection method has been shown to produce results that are compatible with the electrode locations in intraoperative photographs (median disagreement of approx. $3\,\mathrm{mm}$, \cite{Dykstra2012}). 
	
\noindent{\textit{\textbf{Preprocessing.}}}
Extracellular local field potential recordings were preprocessed: CSD  was computed (at $Fl1, Fl2, iP\text{ and }Fb$), $1$--$30\,\mathrm{Hz}$ Fourier filtering ($4^\text{th}$ order Butterworth filter) and subsequent rank normalization was carried out on the recordings. 

\noindent{\textit{\textbf{Model parameters.}}}
Embedding delay was set to $\tau=11$ samples (according to PACF) and embedding dimension was set to $D=7$. The joint manifold ($J'$) was created with $a=1$.

\noindent{\textit{\textbf{Results on multiple seizures and multiple control sections.}}}
We examined the causal connectivity with our method between the fronto-basal ($Fb$), frontal($Fl1$), frontolateral ($Fl2$) and inferior-parietal ($iP$) areas during seizure and interictal control conditions.

Our results showed that during epileptic seizures $iP$ was found to be a dominant driver in the network and also the existence of non-observed drivers could be inferred  (Extended Data \autoref{fig:fotostim_boxplot}). Connectivities calculated on $18$ seizures showed various connection patterns: in $6$ cases $iP$ was a driver in the network, in $1$ case $Fb$ was a driving $Fl1$, in $1$ case $Fb$ and $Fl2$ were bidirectionally coupled and the remaining $10$ cases the nodes were driven by unobserved common cause.

During interictal control conditions $Fb$ was found to be the dominant node (Extended Data \autoref{fig:memo_seizure_microfilm}). Connectivities were calculated on $16$ control periods showed that $Fb$ was playing a driver role in the network. Additionally in $5$ cases $iP$ also drove the other nodes, in $2$ cases $Fl1$ was a driver of $Fl2$, and in many cases bidirectional connections were present between the nodes.

\clearpage
\section{Extended data}\label{si:extended}
\renewcommand{\figurename}{Extended Data Figure}
\setcounter{figure}{0}

\begin{figure}[H]
	\centering
	\includegraphics[scale=0.45]{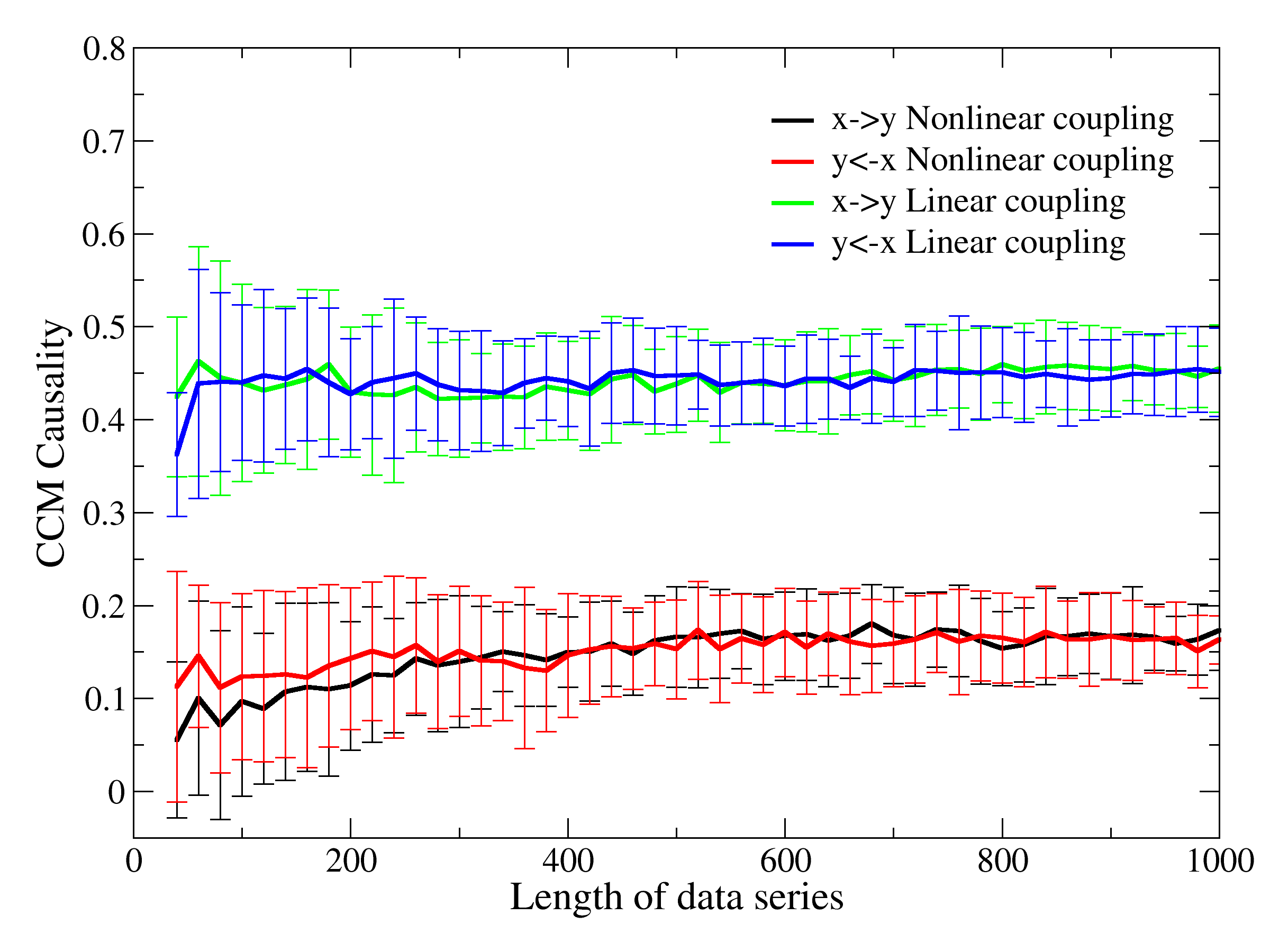}
	\caption{\textit{CCM causality between two observed logistic maps driven by a hidden common cause} The CCM causality measure as a function of the length of the data series are shown. The observed time series (x and y) were not directly coupled, but they were driven by a third logistic map as a common cause. The causality measure shows high values in both directions independently from the data length only if the hidden common cause was linearly coupled to the x and y, thus they were correlated, with $r=0.46$ linear correlation coefficient (blue and green lines). In the case of nonlinear coupling, the CCM measure increases with the increasing data length but converges to lower values, which again corresponds to the linear correlation between x and y $r=0.15$ (black and red lines). Thus, from the data length dependence of CCM, the existence of a hidden common cause can not be inferred reliably.} 
	\label{fig:CommonCauseTest}
\end{figure}

\begin{figure}
	\centering
	\includegraphics[scale=0.45]{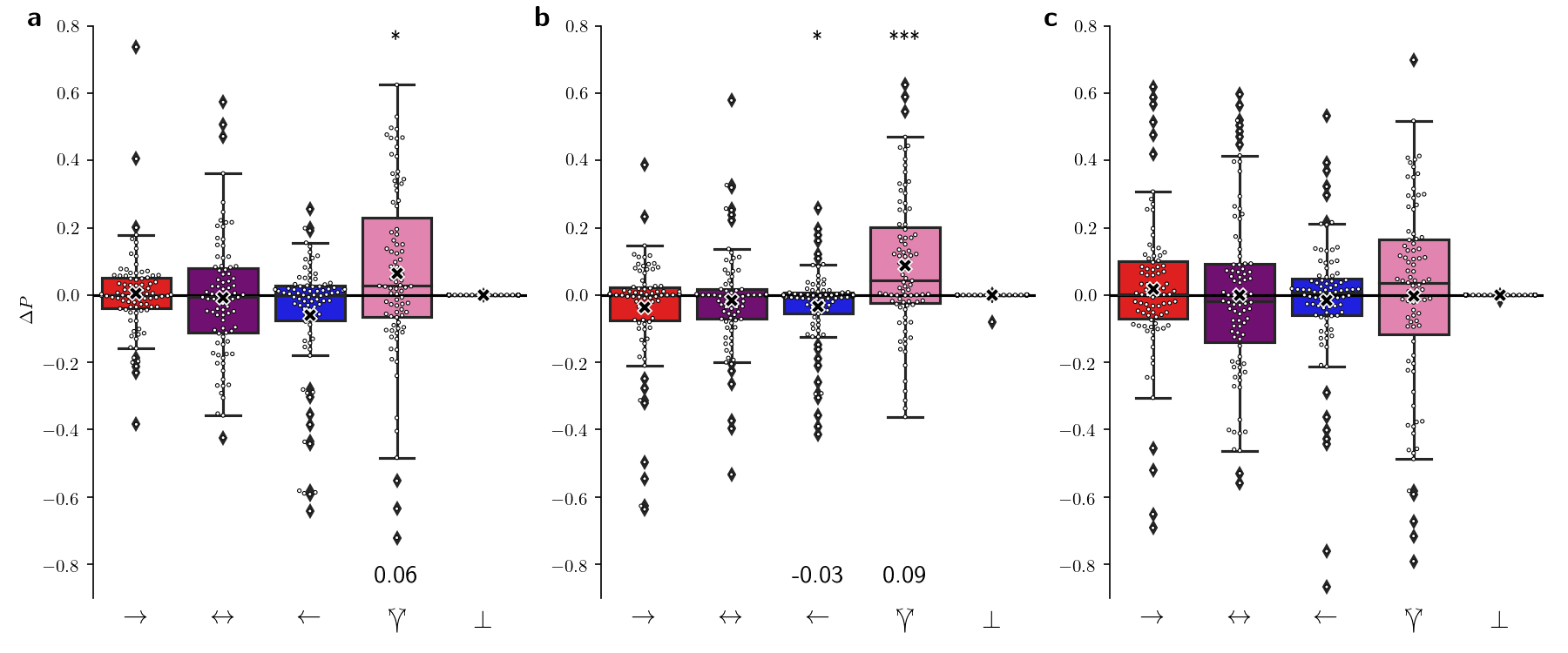}
	\caption{\textit{Differences in probability for causal connection types: stimulus minus control.} The difference in probabilities of causal types are depicted for $P3$--$P4$ (\textbf{a}), $C3$--$C4$ (\textbf{b}), $F3$--$F4$ (\textbf{c}) channel pairs ($n=87$). Boxes denote inter-quartile (IQ) ranges, whiskers are at $1.5$ IQ range from the quartiles and the points outside the whiskers are marked as outliers by diamond symbols. Sample median and mean are indicated by horizontal line and \textbf{x} marker respectively. Asterisks mark significant differences of sample mean (*: $p<0.05$, ***: $p<0.001$). \textbf{a} The probability of common cause significantly increased on the $P3$--$P4$ channel pair ($p=0.026$). \textbf{b} On $C3$--$C4$ channel pair, the probability of common cause shows significant increase ($p=0.0006$) and the directed causal link from $C4$ to $C3$ is significantly decreased ($p=0.011$). \textbf{c} No significant changes in probabilities are detected on the $F3$--$F4$ channel pair.}
	\label{fig:fotostim_boxplot}  
\end{figure}

\begin{figure}
	\centering
	\includegraphics[scale=0.45]{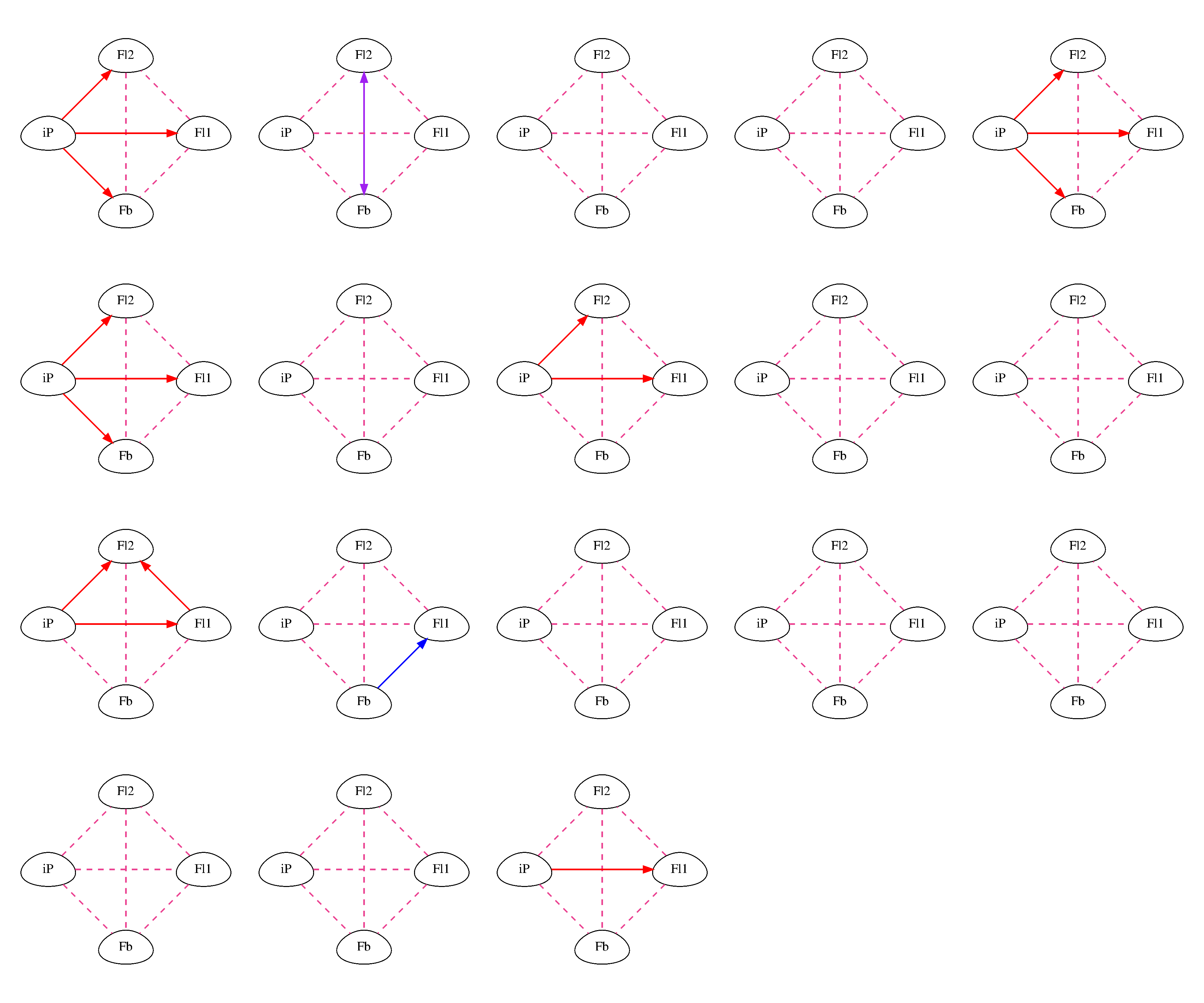}
	\caption{\textit{Inferred causal connectivity during seizures.} Two main types of MAP connectivity pattern were observed: in $10$ cases the existence of a hidden common cause were inferred, while in $6$ cases the infero-parietal ($iP$) area was found as a driver.}
	\label{fig:memo_seizure_microfilm}  
\end{figure}

\begin{figure}
	\centering
	\includegraphics[scale=0.45]{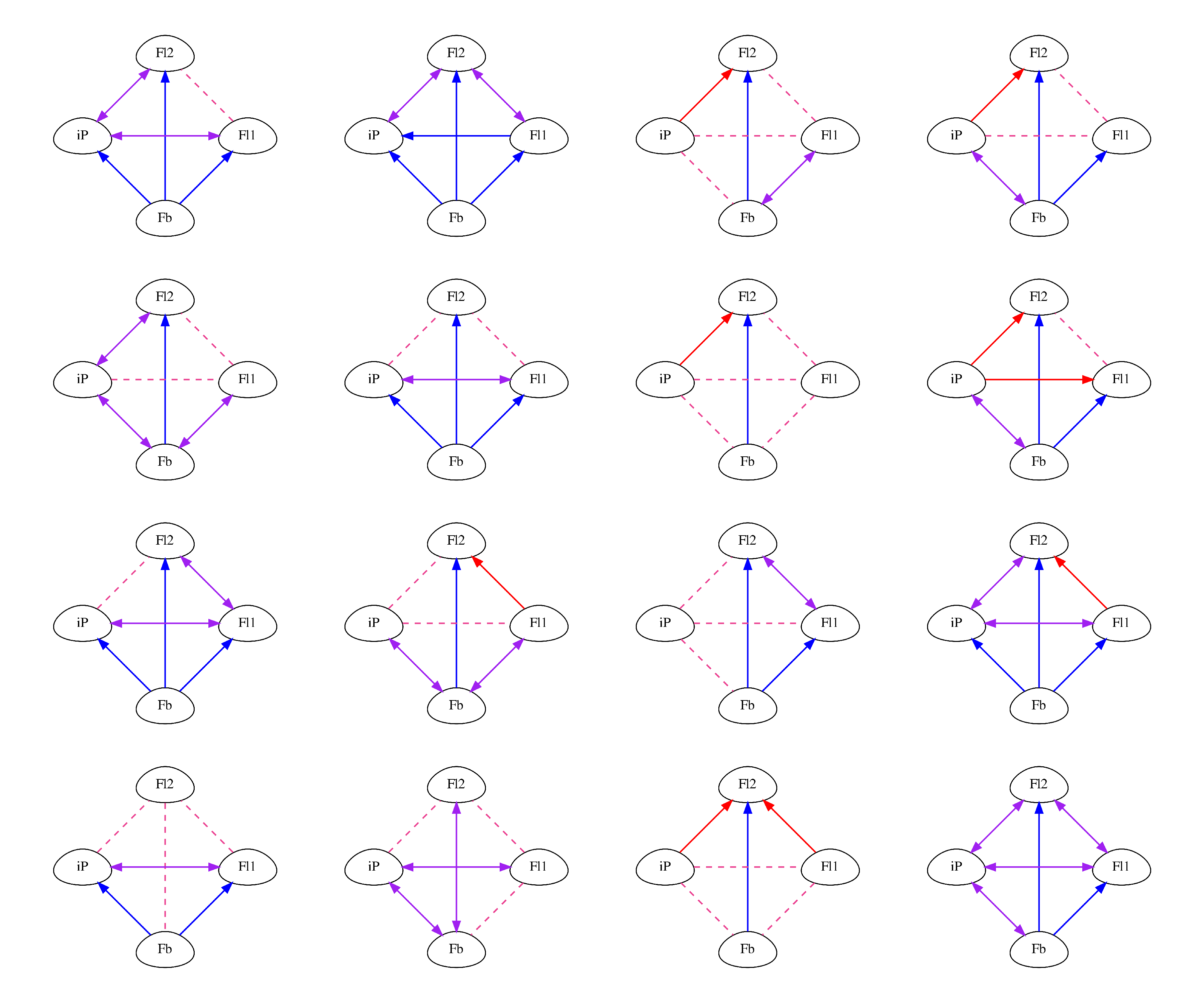}
	\caption{\textit{Causal relations for interictal periods.} MAP connectivity of the $16$ interictal periods show fronto-basal ($Fb$) area as a main driver.}
	\label{fig:memo_control_microfilm}  
\end{figure}

\clearpage



\bibliographystyle{elsarticle-num}

\bibliography{bibs_main}

\vspace{1cm}

\end{document}